%% file: AND_OR.tex
\documentclass[11pt]{article} % 11pt is standard
\usepackage[margin=1in]{geometry} % sets 1 inch margins
\pdfoutput=1 % ensures microtype package works on arXiv

\usepackage[usenames,dvipsnames,svgnames,table]{xcolor} % color

\usepackage{amsmath,amsfonts,amssymb,amsthm}
\usepackage{mathtools}
\usepackage{braket}

\usepackage{enumitem} % control over enumerate/itemize spacing

\usepackage{mmap}
\usepackage[T1]{fontenc}

\usepackage[utf8]{inputenc}
\usepackage{csquotes}
\usepackage[english]{babel}

\usepackage[
    activate={true,nocompatibility},%activate protrusion and expansion
    final, %enable microtype; use "draft" to disable
    tracking=true,
    kerning=true,
    spacing=true,
    factor=1100,%add 10% to the protrusion amount (default is 1000)
    stretch=10,%reduce stretchability (default is 20)
    shrink=10%reduce shrinkability (default is 20)
    ]{microtype}
\microtypecontext{spacing=nonfrench}

\usepackage[backend=biber,% Makes biblatex work correctly
              style=alphabetic,% Usual citation style
        maxbibnames=99,% To not have et al. in references
      minalphanames=3,% To get [BBC+01] instead of [Bea+01]
      maxalphanames=4,% To get [BBBV97] instead of [BBB+97]
            backref=true]{biblatex}

\usepackage{hyperref}
\hypersetup{
    pdftitle={TFNP query separation}, % title
    pdfauthor={}, % author
    colorlinks=true, % false: boxed links; true: colored links
    linkcolor=blue, % color of internal links
    citecolor=blue, % color of links to bibliography
    urlcolor=green % color of external links
}

\input{scripts.tex} % Scripts are in here! %%%%%%%%%%%%%%%%%%%%%%%%%%%
\begin{document}

%\title{On the polynomial degree of total search problems}
\title{Separations in Query Complexity for Total Search
Problems}
\author{
Shalev Ben{-}David\\
\small Institute for Quantum Computing\\
\small University of Waterloo\\
\small \texttt{shalev.b@uwaterloo.ca}
\and
Srijita Kundu\\
\small Institute for Quantum Computing\\
\small University of Waterloo\\
\small\texttt{srijita.kundu@uwaterloo.ca}
}
%\author{Anonymous submission}
\date{}

\maketitle

\begin{abstract}
We study the query complexity analogue of the class
$\mathsf{TFNP}$ of total search problems.
We give a way to convert
partial functions to total search problems under certain settings;
we also give a way to convert search problems back into
partial functions.

As an application, we give new separations for degree-like
measures. We give an exponential separation between
quantum query complexity $\Q(f)$ and approximate degree
$\adeg(f)$ for a total search problem $f$.
We also give an exponential separation between
$\adeg(f)$ and the positive quantum adversary $\Adv(f)$
for a total search problem.

We then strengthen the former separation to upper bound
a larger measure: the two-sided approximate non-negative
degree, also called the conical junta degree.
That is, we obtain a separation of the form
$\Q(f)\gg \appdeg(f)$ for a total search problem $f$,
even though the measure $\appdeg(f)$ is often
larger than $\Q(f)$ and even a separation from
$\R(f)$ was not known. We extend our results
to communication complexity, and obtain an
exponential separation of the form $\QIC(F)\gg \rprt(F)$
for a total search problem $F$, where $\QIC$ is
the quantum information cost and $\rprt$ is the relaxed
partition bound.
Even a weaker separation of the form $\R^{CC}(F)\gg \rprt(F)$
was not known
for total search problems (or even for partial functions).

Most of our separations for total search problems
can be converted to separations for partial functions.
Using this, we reprove the recent exponential separation between
$\Q(f)$ and $\adeg(f)$ for a partial function by
Ambainis and Belovs (2023), among other new results.
%
% First, we reprove the exponential separation between
% quantum query complexity $\Q(f)$ and approximate degree
% $\adeg(f)$ for a partial function, recently shown
% by Ambainis and Belovs (2023). We then show the same
% exponential separation holds for a total search problem.
%
% We further strengthen the separation to upper bound
% non-negative approximate degree, and even its two-sided
% variant; this separation again holds for both partial
% functions and total search problems. We then lift
% this separation to communication complexity
% and obtain an exponential separation between information cost
% $\IC(F)$ and the relaxed partition bound $\rprt(F)$.
% Finally, we give an exponential separation between
% approximate degree and the positive quantum adversary method
% for a total search problem.
\end{abstract}

\clearpage

{\tableofcontents}
\clearpage

\section{Introduction}

\iffalse
\comment{Things to do:
\begin{enumerate}
\item (DONE) TFNP to partial functions, reprove Ambainis-Belovs
\item (DONE) Lifting theorem R->IC?
\item (DONE) Discussing implications for relprt<<IC and relprt<<R for TFNP
\item (FAILS) exact gamma2 norm separation
\item (DONE) Adv<<adeg separation, can come with a discussion about the
Adv vs adeg conjecture (sqrt(cfbs)<=adeg for all relations is easy-ish)
\item (FAILS) Reprove LNNW R<<D for TFNP? Only if there's time
\item (FAILS) Think about whether the above has a deg lower bound
\item Prove relprt << QIC as well
\item (FAILS?) Prove the TFNP to partial function conversion preserves
adeg lower bounds
\item Write intro and clean up paper!!
\end{enumerate}
}
\fi

Query complexity can be studied in several different settings.

\paragraph{Total functions.}
A foundational result in query complexity is that all
``reasonable'' measures of query complexity are polynomially
related to each other for total functions
$f\colon\B^n\to\B$ \cite{BdW02}. The polynomial
relationship of deterministic query complexity $\D(f)$
and randomized query complexity $\R(f)$ follows from
the work of Nisan \cite{Nis91}; the fact that quantum
query complexity $\Q(f)$ and polynomial degree measures
such as $\deg(f)$ and $\adeg(f)$ are also polynomially related
to $\D(f)$ follows from the famous
work of Beals, Buhrman, Cleve, Mosca, and de Wolf \cite{BBC+01}.
Other query complexity measures, such as block sensitivity
$\bs(f)$, fractional block sensitivity $\fbs(f)$,
and certificate complexity $\C(f)$ (which represents
the class $\mathsf{NP}\cap\mathsf{coNP}$) are also polynomially
related to $\D(f)$, as is sensitivity $\s(f)$ \cite{Hua19}.
For a table of polynomial relationships, see \cite{ABK+21}
(as well as some recent improvements in \cite{BBG+21}).

\paragraph{Partial functions.}
This story is radically different in the setting of
\emph{partial functions}, which are Boolean functions $f$
that are defined on only a subset of $\B^n$. Partial
functions can represent more diverse types of problems.
For partial functions, we have a few (usually simple)
relationships such as $\R(f)\le \D(f)$ and 
$\Q(f)\ge\Omega(\sqrt{\bs(f)})$, and essentially any
pair of query measures can be exponentially separated
from each other unless such a separation is forbidden
by these simple relationships. Moreover, the separations
are usually easy to prove: for example, the family of functions
$f_n$ which maps strings of Hamming weight $n/3$
to $0$ and strings of Hamming weight $2n/3$ to $1$
(with all other strings not being in the domain)
has $\R(f_n)=O(1)$ and $\D(f_n)=\Omega(n)$.
The important separation $\Q(f)$ vs.\ $\R(f)$
was first exhibited by Simon around 1994 \cite{Sim97}.
More recently, Ambainis and Belovs \cite{AB23}
gave separations showing $\Q(f)\gg \adeg(f)$
and $\Q(f)\gg \deg(f)$, which were arguably the last major
pairs of measures for which such an exponential separation
was not known.

\paragraph{Total search problems.}
Total search problems are \emph{relations} between
the set of Boolean strings $\B^n$ and a specified
set of partial assignments called certificates,
which specify a small number of values in an $n$-bit string;
an input $x$ is related to a certificate $c$
if $x$ contains $c$, meaning that the values
of the bits specified by $c$ occur in $x$.

The class of total search problems is often considered to be
a query complexity version of $\mathsf{TFNP}$.
For example, a separation between $\R(f)$ and $\Q(f)$ for
total search problems amounts to a separation
between the query versions of $\mathsf{TFNP}\cap\mathsf{FBPP}$
and $\mathsf{TFNP}\cap\mathsf{FBQP}$.
% In this work, we focus on the question:
% which pairs of measures can be separated in $\mathsf{TFNP}$?

Formally, separations for partial functions and
for $\mathsf{TFNP}$ are incomparable, with neither one implying
the other. In practice, however, partial function separations
tend to be much easier, and have historically happened much earlier,
than the corresponding $\mathsf{TFNP}$ separations.
For example, a $\mathsf{TFNP}$ separation between $\R(f)$
and $\Q(f)$ was open for a long time; it follows
from the breakthrough of Yamakawa and Zhandry \cite{YZ22},
but not from any earlier quantum query speedup.
In contrast, the partial function separation follows from
Simon's algorithm or Shor's algorithm \cite{Sim97,Sho97},
among others. Even separating
$\D(f)$ and $\R(f)$ for total search problems is highly nontrivial:
this was first done by \cite{LNNW95} using a construction that involves
expander graphs.

In practice, then, total search problems can be viewed
as an intermediate model between total functions and partial
functions. In this work, we make some formal connections between
partial function and $\mathsf{TFNP}$ separations,
giving tools for converting between the two.
We also give a variety of new separations in both
models, particularly for measures related to polynomial degree.

\subsection{Our results}

We employ two tricks, one to convert partial functions
to total search problems, and another to convert total
search problems to partial functions.

\begin{trick}[Informal, \cite{YZ22}]\label{trick:first}
Suppose we have a partial function such that it is easy
for one model to find a certificate but hard on average
for another model. If the second model has a direct product
theorem, then using $k$-wise independent functions
we can construct a $\mathsf{TFNP}$ problem
which separates the two models.
\end{trick}

This trick was used by \cite{YZ22} for a side result:
the trick was not needed for their
random oracle separation between randomized and quantum
algorithms for search problems, but it was needed
in order to show that the quantum algorithm
can achieve perfect soundness in such a separation.

We generalize this trick to show other novel separations.
We combine it with a second trick for converting back
to a partial function.

\begin{trick}[Informal, see \thm{TFNPtoPartial}]\label{trick:second}
Suppose we have a query $\mathsf{TFNP}$ problem
which separates two models. Using an indexing
construction (similar to cheat sheets \cite{ABK16}),
we can convert this to
a partial function separation between the models
(for many pairs of models).
\end{trick}

We use these tricks to obtain several new separations
in query and communication complexity, both for partial
functions and for $\mathsf{TFNP}$. Our applications
involve polynomial measures such as approximate degree,
approximate non-negative degree (also called conical junta degree),
together with their communication-complexity analogues
(approximate (non-negative) rank and the relaxed partition bound).
For such measures, there is an important subtlety in the way
they are defined: a polynomial $p$ computing a partial
function $f$ can be required to be \emph{bounded},
meaning that $|p(x)|\le 1$ even for strings $x$ outside
of the promise of $f$, or it can be left unbounded.
This choice gives rise to two different versions
of every polynomial measure.

Generally, the unbounded
versions are much stronger computational models
(i.e. much weaker as lower bound techniques) and are
easy to separate from algorithmic measures partial functions.
The bounded versions tend to cling close to the corresponding
algorithmic measures, making them difficult to separate
even for partial functions. For example,
the first exponential separation of quantum query complexity $\Q(f)$
from approximate degree $\adeg(f)$ for a
partial function was only shown in 2023 \cite{AB23},
and a partial function separation of non-negative approximate degree
from randomized query complexity $\R(f)$ was open until this work.

One remarkable property of the above tricks is the following
observation.

\begin{observation}[Informal]
If we apply both \trk{first} and \trk{second}
convert a partial function to $\mathsf{TFNP}$ and
back again,
unbounded degree measures turn into bounded degree measures.
This means we can turn separations for unbounded degree
measures (which are often trivial)
into separations for bounded degree measures
(which are usually very difficult).
\end{observation}

We use the above to prove the following results.

\begin{theorem}\label{thm:main}
Consider the measures quantum query complexity $\Q(f)$
and bounded, two-sided non-negative degree $\appdeg(f)$. Then
\begin{enumerate}
\item There is a partial function exponentially separating
    these measures, with $\Q(f)$ exponentially larger than
    $\appdeg(f)$.
\item There is also a $\mathsf{TFNP}$
problem exponentially separating these measures.
\end{enumerate}
\end{theorem}

This result is quite strong, since
$\appdeg(f)$ (corresponding to the class
$\mathsf{WAPP}\cap\mathsf{coWAPP}$)
is a measure that is usually
equal to $\R(f)$, a measure larger than $\Q(f)$.
We immediately get the following corollaries.

\begin{corollary}
There is a partial function $f$ exponentially separating
$\Q(f)$ and $\adeg(f)$ %(reproving a result of \cite{AB23}),
as well as a $\mathsf{TFNP}$ problem exponentially separating
these two measures.% (a new result).
\end{corollary}

\begin{corollary}
There is a partial function $f$ exponentially separating
$\R(f)$ and bounded two-sided non-negative degree
$\appdeg(f)$, as well as a $\mathsf{TFNP}$ problem
exponentially separating these two measures.
\end{corollary}

Except for the partial function separation of $\Q(f)$
and $\adeg(f)$ (due to Ambainis and Belovs \cite{AB23}),
these corollaries are already new, even
though they are merely special cases of \thm{main}.

By employing a lifting theorem \cite{GPW20,CFK+19}, we also get separations
in communication complexity.

\begin{theorem}\label{thm:communication}
There is a partial communication function $F$
exponentially separating information cost $\IC(F)$ and
the (bounded) relaxed partition bound $\rprt(F)$.
There is also a $\mathsf{TFNP}$ problem exponentially
separating quantum information cost $\QIC(F)$ and $\rprt(F)$.
\end{theorem}

We make a few comments regarding \thm{communication}.
First, we note that it was open even to separate
$\rprt(F)$ from randomized communication complexity, let alone
information cost or quantum information cost. The measure $\rprt(F)$ is also
essentially the largest lower bound technique for
information cost $\IC(F)$ (at least among the ``linear program''
or ``rectangle based'' techniques), so separating
the two is a natural question. Moreover, $\rprt(F)$ is often larger than quantum communication complexity, which makes our separation between $\QIC$ and $\rprt$ even stronger.

The celebrated result of \cite{CMS20}
gave a total function separation between $\IC(F)$ and
one-sided non-negative approximate rank, but the
non-negative approximate rank of the negated version of their
function, $\overline{F}$, was large. In contrast,
our separation upper bounds $\rprt(F)$, which is
equal to two-sided non-negative approximate rank,
but we use either a partial function $F$ or a total
search problem instead of a total function.
One interpretation of our $\mathsf{TFNP}$ separation
can be viewed as saying ``the quantum non-negative approximate logrank
conjecture is strongly false for search problems,''
where by strongly false we mean that even the two-sided
non-negative approximate rank is exponentially smaller
than the quantum communication complexity.

Finally, we examine lower bounds for polynomial measures in
$\mathsf{TFNP}$. We prove the following separation.

\begin{theorem}\label{thm:adv}
There is a query $\mathsf{TFNP}$ problem for which
$\adeg(f)$ is exponentially larger than 
the positive quantum adversary method $\Adv(f)$.
\end{theorem}

We also show several results relating the positive
adversary $\Adv(f)$ to some classical measures,
including variants of block sensitivity.
We show that $\Adv(f)\le\CAdv(f)\le2\Adv(f)^2$,
where $\CAdv(f)$ is the classical adversary method
\cite{Aar06,AKPV18}. We also show that $\CAdv(f)$
is lower bounded by (fractional) block sensitivity
and critical block sensitivity measures.
These results mimic those of \cite{ABK21},
but generalized from partial functions to relations.
One result of \cite{ABK21} fails to generalize:
it no longer seems true that $\CAdv(f)$ is at most
the fractional critical block sensitivity,
and this breaks some of the downstream relationships
such as $\adeg(f)\ge\Omega(\sqrt{\Adv(f)})$
(which holds for partial functions $f$).
% and therefore, the result 
% \[\adeg(f)=\Omega(\sqrt{\CAdv(f)})=\Omega(\sqrt{\Adv(f)}),\]
% which holds for partial functions, might fail for
% search problems (or for relations more generally).
% We leave the problem of relating or separating
% the polynomial and positive adversary methods
% for search problems as an interesting open problem.
% See \sec{discussion} for details.

\subsection{Our techniques}

\paragraph{The first trick.}
\trk{first}, first used in \cite{YZ22}, works as follows.
Suppose we have a partial function $f$ whose promise is
large: the function $f$ is defined on most inputs.
Moreover, suppose that one model, say quantum algorithms,
can quickly find a certificate for $f$ on the inputs in
the promise, while another model, say randomized
algorithms, takes many queries even when the input
is known to come from the uniform distribution.
We wish to make the function \emph{total}: we want
to somehow ensure that the quantum algorithm can find
a certificate on all inputs.

One idea is to have the quantum algorithm ``rerandomize
the input'': instead of trying to find a certificate for
the true input $x$ (which may be impossible), the algorithm
finds a certificate for $x\oplus h$, where $h$ is a random
string and $\oplus$ is the bitwise XOR. Then the algorithm
outputs both the certificate and the randomness $h$.
This allows the quantum algorithm to solve all inputs,
but it also makes the problem easy for randomized algorithms
(since a randomized algorithm now has the freedom to pick
$h$, thereby picking its input).

To overcome this,
we make two modifications (following \cite{YZ22}):
first, instead of using a random string $h$, we use
a string sampled from a $k$-wise independent family $\clH$.
If $k$ is larger than the number of queries in the quantum
algorithm, the algorithm cannot tell apart a random string
from a $k$-wise independent string, so it can still solve
the problem. The second modification is that we now
ask the algorithms to solve $t=O(\polylog n)$
separate instances of the problem,
all rerandomized using the same random seed $h\in\clH$.
The quantum algorithm can still do this if it could
solve the original problem with high probability
(this is just a union bound). However, the move from one
copy to $t$ (all with the same random seed $h$)
allows one to lower bound the randomized algorithm via
a direct product theorem: if the inputs are random,
the probability that the algorithm can solve all of them
quickly (for any specific choice of random seed $h$)
is exponentially small in $t$. Since the family $\clH$
is not too large, we can choose $t$ so that even after
union bounding over all $h\in\clH$, the probability
that the randomized algorithm can solve any one of them
is small. This gives us the lower bound,
and it is how \cite{YZ22} converted their $\Q$ vs.\ $\R$
separation to have perfect soundness.

\paragraph{The second trick.}
\trk{second} converts a $\mathsf{TFNP}$
separation into a partial function separation.
If $f$ is a search problem separating two models,
say randomized and quantum algorithms, we can get a partial
function separation by adding another part of the input,
consisting of an array with one bit per possible certificate (this does not increase the input size too much: there aren't too many certificates, since each certificate is short).
The promise is that if the input is $xy$ where $y$ is
the array, then all cells $y_c$ of the array corresponding
to certificates $c$ consistent with $x$ have the same value
(either all are $0$ or all are $1$).

One can show that for many
models of computation, the complexity of this partial function
is closely related to the complexity of the original search
problem. Indeed, an algorithm that can solve the search problem
can be used to find a certificate and look up the corresponding
cell in the array, solving the function; conversely,
an algorithm which solves the function will need to query
at least one cell $y_c$ in the array corresponding
to a valid certificate, so by checking all the certificates
of queried array cells, one can use that algorithm to solve
the original search problem. See \thm{TFNPtoPartial} for
the formal details.

\paragraph{Complications that arise.}
The above tricks sound easy, but in our applications we encounter
several complications. First, we need a direct product theorem
for the larger measure. This direct product theorem must be
distributional, and it must work for relations;
typical direct product theorems in the literature satisfy
neither of these constraints, and we must use
reductions and workarounds.

The second complication is that to get the results
to work in communication complexity, we must use a lifting
theorem. However, there is no lifting theorem for quantum
algorithms, and the lifting theorem for polynomial degree
was only shown for decision problems, not search problems.
We get around this by proving a version of the approximate degree lifting
for search problems, and then using it to lift to gamma-2 norm.
We do this 
%Instead,
%we lift approximate degree to a lower bound on approximate gamma-2 norm,
even though the approximate gamma-2 norm of our final function
will be very small.
The way this works is that we ``lift the pieces
but build the final relation in communication complexity'';
the degree (and approximate gamma-2 norm) of the pieces is large,
but it drops substantially when the pieces are put together into
the final function,
and we can show that quantum communication does not drop.

\paragraph{Separating degree measures from algorithms.}
To prove \thm{main} separating $\appdeg(f)$ and $\Q(f)$,
we apply \trk{first}. To do this, we need to start
with a partial function for which it is easy for polynomials
to find certificates for most inputs, but hard for quantum
algorithm even on average. Our function will be the
OR function on $n^2$ bits composed with an AND function
on $\log(n/2)$ bits, where a certificate is a $1$-input
to one of the $n^2$ AND functions.
On a random input, each AND evaluates
to $1$ with probability $2/n$, so there are around $2n$
valid certificates. Indeed, except with exponentially small
probability, the number of certificates is between $n$ and $3n$
on a random input. We set this to be the promise: that the
number of certificates is in that range (or $0$).

A quantum algorithm can find such a certificate
via Grover search, but this requires $\Omega(\sqrt{n})$
queries, polynomial in the input size of $O(n^2\log n)$.
In contrast, a polynomial can do the following:
first, construct a polynomial of degree $O(\log n)$
for each certificate, which evaluates to $1$ if it
is present and to $0$ otherwise; then,
add up all these polynomials and divide by $3n$.
The resulting polynomial has degree $O(\log n)$,
yet it is bounded on most inputs and correctly computes
the function in the sense that it outputs at least $1/3$
on almost all $1$-inputs but outputs $0$ on the $0$-input.
(This polynomial can be modified to ``output a certificate''
in a suitable sense.) We can amplify this polynomial
to increase its success probability to $1-1/n$ nearly
all inputs. The catch is that this polynomial
is unbounded outside of the promise (if the number
of certificates is $n^2$, the polynomial outputs
$n/3$ before amplification
and something even larger after amplification;
a bounded polynomial should output something in $[0,1]$).

To apply \trk{first}, we need a direct product theorem
for quantum algorithms. Unfortunately, the direct product
theorem of \cite{LR13} is not good enough for our purposes,
as it is not distributional and does not lower bound
the difficulty of finding a certificate if we know
the input is a $1$-input.
Instead, we use a direct product theorem for approximate
degree given by Sherstov \cite{She12}, in combination
with some worst-case-to-average-case and search-to-decision
reductions. The result could also follow from
the search lower bound of \cite{ASdW07,Spa08}
(using the multiplicative weight adversary method)
in combination with some reductions.
% Instead, we use a reduction
% from a search lower bound of \cite{ASdW07,Spa08} to prove
% the necessary direct product theorem for this function,
% together with a worst-case-to-average-case reduction.
% \comment{Change this.}

\paragraph{Lifting to communication.}
In order to get the separation of $\IC(F)$ from $\rprt(F)$
(\thm{communication}), we employ a lifting theorem.
We start with a query separation, and lift it to communication
via composition with the inner product function.
The upper bound for $\rprt(F)$ then follows easily from
the upper bound on the query measure $\appdeg(f)$.
For the lower bound, since there is no lifting theorem
for quantum communication complexity, we lift randomized
communication complexity using \cite{GPW20,CFK+19}.

How do we convert this to a lower bound on information cost,
$\IC(F)$? For this, we use the following trick:
$\IC(F)$ is known to be equal to amortized communication
complexity \cite{BR14}, that is, the communication complexity
of solving $n$ copies of $F$, divided by $n$.
This is not known to equal the randomized communication
complexity of $F$, as there is no ``direct sum'' theorem
in communication complexity. However, there is a direct
sum theorem in query complexity! The operations of
``take $n$ copies of the function'' and ``lift the function
from query to communication'' commute, so the
direct sum theorem in query complexity \cite{JKS10}
combined with the lifting theorem for randomized
communication \cite{CFK+19} and the amortized characterization
of information cost \cite{BR14} give the following theorem.

\begin{theorem}[Informal, see \thm{IC-lifting}]
Randomized query complexity lifts to a lower bound
on information cost. In other words, information cost
and randomized communication cost are equal for lifted functions.
\end{theorem}

Our separation between $\QIC(F)$ and $\rprt(F)$ is trickier, because although the same relationship holds between quantum information cost and amortized communication cost, there is no quantum query-to-communication theorem. To overcome this problem, we instead prove a theorem lifting (distributional) approximate degree to (distributional) approximate $\gamma_2$ norm for search problems. \cite{She11} proved the analogous result for boolean functions, but we show the search problem version can be proved using the techniques from that work as well. Since distributional approximate $\gamma_2$ norm lower bounds distributional quantum communication complexity, using our lifting theorem, we have a direct product theorem for the quantum communication complexity of finding several certificates in a lifted function. After this the communication version of \trk{first} can be used to get a lower bound on amortized quantum communication, and hence $\QIC$.

\paragraph{Approximate degree vs.\ adversary methods and block
sensitivity.} To separate $\adeg(f)$ from $\CAdv(f)$
for a total search problem (\thm{adv}), we could try the same
strategy of starting with a partial function and employing
\trk{first}, using the direct product theorem for $\adeg(f)$
\cite{She12}. However, instead of attempting this, we present
a different, simpler proof.

Our total search problem will simply
be the pigeonhole problem with $2n$ pigeons and $n$ holes:
for $n$ a power of $2$,
we set the input to have $2n$ blocks of $\log n$ bits each,
and a valid certificate is any two identical blocks
(which must exist due to the pigeonhole principle).
The upper bound on the adversary method works by noticing
that any two inputs with no overlapping certificates must
have high Hamming distance; a version of the ``property testing
barrier'' then applies to show that the positive quantum
adversary bound (or even the classical adversary bound)
is small. The lower bound on the approximate degree
works by a reduction from the collision problem:
a polynomial system that can find a certificate for
this pigeonhole problem can be converted into a polynomial
for solving collision after suitable modifications
(including a type of mild symmetrization argument).

We also present other results regarding the relationships
between the positive adversary methods, the block sensitivity
measures (including fractional and critical block sensitivity),
and approximate degree. Most such relations follow via
arguments similar to those for partial functions,
which can be found in \cite{ABK21} (see also
\cite{AKPV18} for total function relations).

\subsection{Open problems}

There are several remaining pairs of measures for which
we do not have a separation for total function search problems.
We list some notable ones.

\begin{open}
Is there a total search problem $f$ exhibiting
an exponential separation between exact quantum
query complexity $\Q_E(f)$ and randomized query complexity
$\R(f)$, in either direction?
\end{open}

Our techniques do not work very well for exact measures,
so we do not know how to handle $\Q_E$.

\begin{open}
Is there a total search problem $f$ exhibiting an
exponential separation between approximate degree
and approximate non-negative degree?
How about exact degree and approximate degree?
\end{open}

\begin{open}
For partial functions, it holds that
$\adeg(f)=\Omega(\sqrt{\Adv(f)})$. Does this hold
for total search problems? If not, can an exponential
separation of the form $\Adv(f)\gg \adeg(f)$ exist?
(In the other direction, we give an exponential separation
in \thm{adv}.)
\end{open}

One of our results gives a TFNP separation
$\QIC(F)\gg \rprt(F)$, but in a setting
where we do not know how to give a partial function separation.

\begin{open}
Is there a partial communication problem $F$
exhibiting a separation $\QIC(F)\gg\rprt(F)$?
\end{open}

Finally, unlike in query complexity, in the setting
of communication complexity relationships for total
decision problems are generally not known. In fact,
the approximate logrank conjecture was shown to be
false even for total decision problems, meaning 
that there is an exponential separation of the form
$\R^{CC}(F)\gg \log\rank_\epsilon(F)$
(and even $\Q^{CC}(F)$ can be shown to be exponenitally
larger than one-sided non-negative approximate logrank)
\cite{ABT19,SdW19,CMS20}. However, such separations
are not known against two-sided non-negative approximate
logrank (that is, relaxed partition bound).

\begin{open}
Is the approximate non-negative logrank conjecture also false?
That is, is there a total decision communication problem
$F$ exhibiting an exponential separation
$\R^{CC}(F)\gg\rprt(F)$,
or even $\QIC^{CC}(F)\gg\rprt(F)$?
\end{open}

Our results settle this problem for total search problems,
but it is plausible that such a separation also holds
for total decision problems.

\section{Preliminaries}

\subsection{TFNP and related definitions}

\paragraph{Relations.}
A relation in query complexity can be defined in two ways.
One way is via a subset $R$ of $\B^n\times \Sigma$ for some
nonempty finite alphabet $\Sigma$ (representing the set
of output symbols). A second way is by assigning each
string in $\B^n$ a set of allowed output symbols in $\Sigma$
via a function $f\colon \B^n\to\mathcal{P}(\Sigma)$,
where $\mathcal{P}$ denotes the power set.
We can define $f$ in terms of $R$ via
$f(x)\colon \{c\in\Sigma:(x,c)\in R\}$, and we can define
$R$ in terms of $f$ via $R=\{(x,c):x\in\B^n,c\in f(x)\}$.
The inputs $x$ with $f(x)=\emptyset$ are considered not to be
in the promise of the relation; the relation is total if
$f(x)\neq\emptyset$ for all $x\in\B^n$.
In general, an algorithm computing $f$ will only be required
to output \emph{some} symbol in $f(x)$ (not the entire set).

When $f$ is not total, we let $\Dom(f)$ denote the set
of strings $x\in\B^n$ with $f(x)\ne\emptyset$,
and we consider $f$ to be a relation between $\Dom(f)$
and $\Sigma$; in effect, we treat $\Dom(f)$ as the promise
and only require algorithms to compute $f$ on inputs
in $\Dom(f)$.

\paragraph{Partial assignments.}
For $n\in\bN$, a partial assignment $p$ on $n$ bits
is an element in $\{0,1,*\}^n$. This represents partial
knowledge of a string in $\B^n$. We identify a partial
assignment $p$ with the set $\{(i,p_i):i\in[n],p_i\ne *\}$.
This allows us to write $|p|$ for the number of non-$*$ bits
of $p$ (the \emph{size} of the partial assignment),
and it also allows for notation such as $p\subseteq x$
to denote that $p$ is consistent with a string $x\in\B^n$
(i.e.\ $p\subseteq x$ happens if and only if $p$ and $x$
agree on the non-$*$ bits of $p$).

\paragraph{Search problems.}
Let $n\in\bN$. We define a search problem on $n$ bits
to be a set $S\subseteq\{0,1,*\}^n$ of partial assignments
called \emph{certificates}.
The certifying cost of $S$ is defined to be
$\cert(S)\coloneqq \max_{c\in S} |c|$. We will generally
be interested in a family of search problems $\{S_n\}$,
with $S_n$ defined on $n$ bits, such that
$\cert(S_n)=O(\polylog n)$. We consider this requirement
to be part of the definition of a family $\{S_n\}$ of search
problems.

We identify a search problem $S$ defined as above with
the function $f_S\colon \B^n\to\mathcal{P}(S)$
defined by $f_S(x)\coloneqq\{c\in S:c\subseteq x\}$;
that is, $f_S(x)$ returns the set of all certificates $c$ in
our certificate collection $S$ which are
consistent with the input $x$. This function $f_S$
specifies a relation between the inputs $x\in\B^n$
and the outputs $c\in S$, with $x$ and $c$ related if
and only if $c\subseteq x$; here $S$ is the output alphabet.
In other words, the task
$f_S$ corresponding to the collection $S$ is to find
a certificate $c\in S$ consistent with the given input $x$.
We say the search problem is total if the defined relation
is total (that is, if
every input in $\B^n$ contains at least one certificate in $S$).

Formally, we define (the query version of) $\mathsf{TFNP}$
as the set of families of total search problems: $\mathsf{TFNP}$
contains all families $\{f_{n_i}\}_{i\in\bN}$ of relations
$f_{n_i}$ on $n_i$ bits such that:
\begin{enumerate}
\item For each $i\in \bN$,
$f_{n_i}$ arises from a search problem, i.e.\ $f_{n_i}=f_S$
for some set of certificates $S\subseteq\{0,1,*\}^{n_i}$
\item For each $i\in\bN$, the relation $f_{n_i}$ is total
\item The input sizes $n_i$ strictly increase with $i$
\item The certificate sizes satisfy $\cert(S_{n_i})=O(\polylog n_i)$,
where $S_{n_i}$ is the set of certificates corresponding
to $f_{n_i}$.
\end{enumerate}

% Let $f_S$ be a search problem on $ D \subseteq \{0,1\}^n$ with output set $S$. For each $s\in f_S(x)$, there is a minimum certificate $c_s$ certifying that $s$ is a valid output for $x$; that is, $s\in f_S(x)$ for all $x\in D$ such that $c \subseteq x$. Let $C(f_S)$ denote the certificate complexity of $f_S$: the maximum size of a certificate $c_s$ for $s\in S$. If $f_S$ is a TFNP problem, then $D=\{0,1\}^n$, and its certificates are polylogarithmic-sized: we can in fact let the certificates be the outputs. In this case, $S$ is a set of certificates, and $f_S(x) := \{c\in S: c \subseteq x\}$ is nonempty for every $x\in \{0,1\}^n$.

\paragraph{Other variants.}
We note that another natural definition of a version of
$\mathsf{TFNP}$ would be to say that a family of relational
problems $f_n$ has short certificates. That is, for a relation $f$,
given $x\in\B^n$ and an output symbol $s\in f(x)$,
we say that a partial assignment $c\in\{0,1,*\}^n$ is
a certificate for $(x,s)$ with respect to $f$ if $c\subseteq x$
and for all $y\in\B^n$ with $c\subseteq y$, we have $s\in f(y)$.
We define
the certificate complexity of $(x,s)$ with respect to $f$
(denoted $\C_{x,s}(f)$) to be the minimum size $|c|$ of a certificate
$c$ for $(x,s)$, we define
$\C_x(f)\coloneqq \min_{s\in f(x)} \C_{x,s}(f)$,
and we define $\C(f)=\max_{x\in\B^n}\C_x(f)$.
We can then say that a family of total relations $f_n$ has small
certificates if $\C(f_n)=O(\polylog n)$.

The difference between
this definition and the previous one is that if we know
that $f$ has short certificates, this does not force an algorithm
computing $f$ to \emph{find} such a certificate. Our definition
of search problems is designed to force an algorithm for $f$
to find a certificate for $f$; a total search problem
is therefore more restrictive a criterion than simply saying
that $f$ has short certificates. We will use the more restrictive
definition in this work (see the definition of
$\mathsf{TFNP}$ in query complexity, given above).

\subsection{Query measures and their relational versions}

\subsubsection*{Deterministic, randomized, and quantum algorithms}
Defining deterministic, randomized, and quantum query complexities
for relations works in the usual way; see \cite{BdW02} for example.
The main difference is that the outputs of the algorithms
will be symbols in the alphabet $\Sigma$ (instead of just in $\B$),
and given input $x\in\B^n$, the algorithm is said to correctly
compute $f(x)$ if its output is in the set $f(x)$
(instead of being equal to the symbol $f(x)$, as happened
when $f$ represented a function). If $f$ is not total,
we will only require the algorithm to work on inputs $x$
that are in the promise set $\Dom(f)$,
since outputting a symbol in $f(x)$ is impossible when
$f(x)=\emptyset$.

One special thing that happens for search problems is that
the difference between bounded-error (``Monte Carlo'') algorithms
and zero-error (``Las Vegas'') algorithms disappear. This is because
a correct output to a search problem is a certificate,
and the certificates are required to be small (polylog-sized),
so it is always possible for an algorithm to check the certificate
before outputting it. To convert a bounded-error algorithm
to a zero-error algorithm, just modify the algorithm to check
the certificate before it is returned as output,
and if it happened to be wrong, rerun the algorithm from
scratch using fresh randomness.

The above gives us definitions of $\D(f)$
(deterministic query complexity), $\R(f)$
(randomized query complexity), $\Q(f)$ (quantum query complexity),
and $\Q_E(f)$ (exact quantum query complexity, representing
algorithms which make no errors and always terminate after a fixed
number of queries). However, the measures $\R_0(f)$ and
$\Q_0(f)$ (for expected number of queries made by Las Vegas randomized and quantum algorithms) will not be used as they collapse with
$\R(f)$ and $\Q(f)$ respectively (up to constant factors
and additive $\polylog(n)$ terms).

\subsubsection*{Degree measures}

To define polynomial degree, instead of representing
a Boolean function by a polynomial we will represent
it by a collection of polynomials $\{p_s\}_{s\in\Sigma}$,
one for each output symbol.
The polynomial's output $p_s(x)$ intuitively represents
the probability that an algorithm outputs the symbol
$s$ when run on input $x$.

\begin{definition}[Approximate degree for relations]
Let $f$ be a relation on $n$ bits with output alphabet
$\Sigma$. We say that a collection of real $n$-variate polynomials
$\{p_s\}_{s\in\Sigma}$ computes $f$ to error $\eps$
if the following conditions hold:
\begin{enumerate}
\item $p_s(x)\ge 0$ for all $s\in \Sigma$ and all $x\in\B^n$
\item $\sum_{s\in\Sigma} p_s(x)\le 1$ for all $x\in\B^n$
\item $\sum_{s\in f(x)} p_s(x)\ge 1-\eps$ for all $x\in\Dom(f)$.
\end{enumerate}
The degree of the collection $\{p_s\}_{s\in\Sigma}$ is defined
as the maximum degree of any polynomial $p_s$ in the collection.
The $\eps$-approximate degree $\adeg_\eps(f)$ of
the relation $f$ is the minimum degree of a collection of
polynomials computing $f$ to error $\eps$.

When $\eps=1/3$, we omit it and write $\adeg(f)$.
When $\eps=0$, we instead write $\deg(f)$, which is called
the exact polynomial degree of $f$.

We will sometimes use $\adeg(f,\mu)$ for approximate degree with respect to a specific distribution $\mu$ on the inputs in $\Dom(f)$, which will be the same as the above definition, except with the last constraint replaced with $\sum_{x\in \Dom(f)}\mu(x)\sum_{s\in f(x)}p_s(x) \geq 1-\eps$.
\end{definition}

% \begin{definition}[Approximate degree for search functions]
% We say a set of real polynomials $\{p_s\}_{s\in S}$, where each $p_s$ is on defined on $n$ variables, $\eps$-approximates $f_S: D\to \{0,1\}$ if:
% \begin{enumerate}
% \item $p_s(x_1\ldots x_n) \in [0,1] \, \forall s\in S, x \in \{0,1\}^n$
% \item $\sum_{s\in S} p_c(x_1\ldots x_n) \leq 1 \, \forall x \in \{0,1\}^n$
% \item $\sum_{s \in f_S(x)} p_s(x_1\ldots x_n) \geq 1-\eps \, \forall x \in D$.
% \end{enumerate}
% The degree of a set of approximating polynomials is the maximum degree of a polynomial in the set. The $\eps$-approximate degree $\adeg_\eps(f_S)$ of $f_S$ is the minimum degree of a set of polynomials that $\eps$-approximates $f_S$.
% \end{definition}

This definition only requires the polynomials to correctly
compute $f(x)$ on inputs in the promise $\Dom(f)$, but
it requires the boundedness conditions $\sum_s p_s(x)\le 1$
and $p_s(x)\ge 0$
even for inputs $x\in\B^n$ outside the promise. This is analogous
to the usual way of defining approximate degree of partial
functions; the idea is that if $p_s(x)$ represents the probability
of outputting symbol $s$ when given input $x$, then even
if the input $x$ is invalid, the probabilities should make sense.

This brings up another subtlety: another way of defining
polynomial degree is to replace the second condition with
$\sum_s p_s(x)=1$ (i.e. require the probabilities to sum
to $1$ instead of to at most $1$). Such a change turns out
not to matter, because if $\sum_s p_s(x)\le 1$, then the polynomial
$p_{\bot}(x)=1-\sum_s p_s(x)$ is non-negative for $x\in\B^n$;
we can then view $p_{\bot}(x)$ as representing the probability
of a special output symbol $\bot$. If we wish, we could
make the polynomial output another symbol $s\in\Sigma$
instead of $\bot$ by replacing $p_s(x)$ with $p_s(x)+p_{\bot}(x)$.
Doing so will ensure that $\sum_s p_s(x)=1$.
Usually, we prefer to treat $\bot$ as a separate symbol;
we will sometimes reason about the collection
of polynomials $\{p_s\}_{s\in\Sigma\cup\{\bot\}}$
instead of $\{p_s\}_{s\in\Sigma}$.

We will occasionally also need a version of polynomial
degree for which the polynomials are not required to be
bounded outside of the promise. We introduce
the following definition.

\begin{definition}[Unbounded approximate degree for relations]
Let $f$ be a relation on $n$ bits with output alphabet
$\Sigma$. We say that a collection of real $n$-variate polynomials
$\{p_s\}_{s\in\Sigma}$ unboundedly computes $f$ to error $\eps$
if the following conditions hold:
\begin{enumerate}
\item $p_s(x) \geq 0$ for all $s\in \Sigma, x \in \Dom(f)$
\item $\sum_{s\in \Sigma} p_s(x) \leq 1$ for all $x\in \Dom(f)$
\item $\sum_{s \in f(x)} p_s(x) \geq 1-\eps$ for all $x\in\Dom(f)$.
\end{enumerate}
The degree of a collection of unbounded approximating polynomials is the maximum degree of a polynomial in the set. The unbounded $\eps$-approximate degree $\udeg_\eps(f)$ of $f$ is the minimum degree of a collection of polynomials that unboundedly computes $f$ to error $\eps$.
\end{definition}

\begin{definition}[Two-sided non-negative degree, or conical junta degree, for relations]
A conjunction on $n$ variables $x_1,\ldots, x_n$ is a monomial which is a product of $x_i$-s or $(1-x_i)$-s; it takes values 0 or 1 for all $x\in \B^n$. A conical junta is a non-negative linear combination of conjunctions, and its conical degree is the same as its degree as a polynomial.

Let $f$ be a relation on $n$ bits with output alphabet $\Sigma$. We say a collection of $n$-variable conical juntas $\{h_s\}_{s\in \Sigma}$ computes $f$ up to error $\eps$ if the following conditions hold:
\begin{enumerate}
\item $\sum_{s\in \Sigma}h_s(x) \leq 1$ for all $x\in\B^n$
\item $\sum_{s\in f(x)}h_s(x) \geq 1-\eps$ for all $x\in \Dom(f)$.\footnote{We don't require the $h_s(x) \geq 0$ condition unlike polynomials, because conical juntas are always positive for $x\in \B^n$.}
\end{enumerate}
The degree of $\{h_s\}_{s\in \Sigma}$ is defined as the maximum degree of any conical junta in the set. The $\eps$-approximate conical junta degree $\adeg\!{}_\eps^{+2}(f)$
of $f$ is the minimum degree of a collection of conical juntas that computes $f$ to error $\eps$. When $\eps=1/3$, we omit it and write $\appdeg(f)$.
\end{definition}

\subsubsection*{Sensitivity and adversary measures}
We start by defining several variations of adversary bounds
and block sensitivity measures. For a relation $f$,
we use the notation $\Delta(f)$ to denote the set
of pairs with disjoint output symbols:
$\Delta(f)\coloneqq\{(x,y):f(x)\cap f(y)=\emptyset\}$.

\begin{definition}
Let $f$ be a relation on $n$ bits. The (positive-weight)
quantum adversary bound, denoted $\Adv(f)$,
is defined as the optimal value of
\begin{equation*}
\begin{array}{lrll}
\text{min}  & T& &\\
\text{s.t.}& \displaystyle\sum_{i=1}^n m_x(i) &\le T &\forall x\in\Dom(f)\\
           &  \displaystyle\sum_{i:x_i\ne y_i}\sqrt{m_x(i)m_y(i)}&\ge 1    &\forall (x,y)\in \Delta(f)\\
           & m_x(i)&\ge 0 &\forall x\in\Dom(f),i\in[n].
\end{array}
\end{equation*}
\end{definition}
Intuitively, the values $m_x(i)$ represent ``how much
the algorithm queries position $i$ when it is run on input
$x$''. The objective value $T$ is the total number
of queries made on any input $x$. When
$f(x)\cap f(y)=\emptyset$, it means the quantum
algorithm must distinguish $x$ and $y$, and that forces
it to query a bit where they disagree (to some extent,
at least). A well-known result in query complexity
(see \cite{SS06} for the case where $f$ is a function)
gives
\[\Q(f)\ge\Omega(\Adv(f)).\]

We next define the classical adversary method.

\begin{definition}
Let $f$ be a relation on $n$ bits. The classical
adversary bound, denoted $\CAdv(f)$,
is defined as the optimal value of
\begin{equation*}
\begin{array}{lrll}
\text{min}  & T& &\\
\text{s.t.}& \displaystyle\sum_{i=1}^n m_x(i) &\le T &\forall x\in\Dom(f)\\
           &  \displaystyle\sum_{i:x_i\ne y_i}\min\{m_x(i),m_y(i)\}&\ge 1    &\forall (x,y)\in \Delta(f)\\
           & m_x(i)&\ge 0 &\forall x\in\Dom(f),i\in[n].
\end{array}
\end{equation*}
\end{definition}
This definition is the same except the square root
is replaced by a minimum. The classical adversary
method lower bounds randomized query complexity \cite{Aar06,AKPV18,ABK21}
\[\R(f)\ge \Omega(\CAdv(f)).\]
It is also easy to see that $\CAdv(f)\ge \Adv(f)$
for all relational problems $f$, since $\Adv(f)$ is a relaxation
of the minimization problem of $\CAdv(f)$.

We define two versions of block sensitivity,
and two versions of fractional block sensitivity.

\begin{definition}[Block sensitivity]
Let $f$ be a relation on $n$ bits and let $x\in\Dom(f)$.
We call a set $B\subseteq[n]$ a \emph{block},
and we denote by $x^B$ the string $x$ with the bits
in $B$ flipped.

The block sensitivity of $f$ at $x$,
denoted $\bs_x(f)$, is the maximum number $k$
such that there are $k$ blocks 
$B_1,B_2,\dots,B_k\subseteq[n]$ which are disjoint
(i.e. $B_i\cap B_j=\emptyset$ if $i\ne j$) and
which are all sensitive for $x$, meaning that
$f(x^{B_i})\cap f(x)=\emptyset$ for all $i\in[k]$.

The block sensitivity of $f$ is defined as
$\bs(f)\coloneqq \max_x\bs_x(f)$.
\end{definition}

\begin{definition}[Critical block sensitivity]
 We call an input $x\in\Dom(f)$ \emph{critical}
if $|f(x)|=1$ (i.e.\ there is only one valid output
for input $x$). Call a function $f'$ a \emph{completion}
of $f$ if $f'(x)\in f(x)$ for all $x\in\Dom(f)$;
that is, $f'$ is a function whose outputs are always
legal outputs for the relation $f$.

The \emph{critical block sensitivity} of $f$ is defined as
\[\cbs(f)\coloneqq \min_{f'}\max_{x\text{ crit}}\bs_x(f'),\]
where $f'$ ranges over all completions of $f$ and
where $x$ ranges over all critical inputs to $f$.
Note that this definition uses the block sensitivity
of the function $f'$, but only for inputs critical for
the relation $f$.
\end{definition}

% \begin{definition}[Strong block sensitivity]
% Let $f$ be a relation and let $x\in\Dom(f)$.
% We say a block $B\subseteq[n]$ is \emph{semisensitive}
% for $x$ (with respect to the relation $f$) if
% there exist $y,z\in\Dom(f)$ such that
% $f(y)\cap f(z)=\emptyset$ and both $y$ and $z$
% agree with $x$ outside of $B$:
% $x_i=y_i=z_i$ for all $i\in[n]\setminus B$.

% The strong block sensitivity of $f$ at $x$,
% denoted $\sbs_x(f)$, is the maximum number $k$
% such that there are $k$ disjoint semisensitive
% blocks for $x$. The strong block sensitivity of $f$
% is $\sbs(f)\coloneqq\max_x\sbs_x(f)$.
% \end{definition}

\begin{definition}[Fractional versions]
The \emph{fractional block sensitivity} $\fbs_x(f)$
of a relation $f$ at an input $x$ is the maximum
total amount of non-negative weight, $\sum_B w(B)$,
which can be assigned to sensitive blocks $B$ of $x$
under the condition that the blocks containing
an index $i\in[n]$ must have weight at most $1$:
$\sum_{B\ni i}w(B)\le 1$. The fractional
block sensitivity of $f$ is $\fbs(f)\coloneqq\max_x\fbs_x(f)$.

The fractional critical block sensitivity is defined
the same way as $\cbs(f)$, except with $\fbs_x(f')$
instead of $\bs_x(f')$ as the underlying measure:
\[\fcbs(f)\coloneqq \min_{f'}\max_{x\text{ crit}}\fbs_x(f').\]
%
% The fractional strong block sensitivity is defined
% the same way as fractional block sensitivity,
% except the blocks $B$ that are assigned weight
% are the semisensitive blocks of $x$ instead of the sensitive
% blocks of $x$. We denote it by $\fsbs(f)$.
\end{definition}

We note that the linear program for fractional
block sensitivity at a fixed input $x$ has a linear
program
\begin{equation*}
\begin{array}{lrll}
\text{max}  & \sum_B w(B)& &\\
\text{s.t.}& \displaystyle\sum_{B\ni i} w(B) &\le 1 &\forall i\in[n]\\
           &  w(B) &\ge 0 &\forall B,
\end{array}
\end{equation*}
where the only dependence on $f$ and $x$ is in the allowed
set of blocks $B\subseteq[n]$, which must be sensitive for $x$
with respect to $f$.
% (they must be sensitive
% for $x$ with respect to $f$
% in the case of $\fbs_x(f)$, and they must be semisensitive
% in the case of $\fsbs_x(f)$). 
This linear program
has a dual, which is also equal to fractional block
sensitivity:
\begin{equation*}
\begin{array}{lrll}
\text{min}  & \sum_{i\in[n]} c(i)& &\\
\text{s.t.}&  \sum_{i\in B} c(i) &\ge 1 &\forall B \\
           & c(i) &\ge 0 &\forall i\in[n].
\end{array}
\end{equation*}
The weights $c(i)$ are assigned to the indices $i\in[n]$;
since they depend on $x$, we often denote these weights
by $c_x(i)$. A feasible solution to this minimization
problem is called a fractional certificate;
it has the property that for all sensitive
%(or semisensitive)
$B$, we have $\sum_{i\in B}c_x(i)\ge 1$, which means
that $c_x$ ``detects if a sensitive block has been flipped''.
The minimum amount of weight required to do this
for an input $x$ is equal to $\fbs_x(f)$, and the minimum
required to do this for all $x$ is equal to $\fbs(f)$.
In other words, ``fractional certificate complexity
is equal to fractional block sensitivity''.

\section{Converting search problems to partial functions}

In this section, we show how to get separation for partial functions, if we already have a corresponding separation for a total search problem. 
%This will let us reprove the separation between
%$\adeg(f)$ and $\Q(f)$ for partial functions that was recently given by \cite{AB23}.

\begin{definition}
Let $f$ be a query $\mathsf{TFNP}$ problem on $n$ bits.
The partial function associated with $f$, denoted
$f_{\fcn}$, is defined as follows. The inputs are strings
in $\B^n\times \B^m$, where $m=|S|$ is the number of certificates
that can be outputs of $f$. Use the certificates $c\in S$
to index the second part of the string. For $x\in\B^n$ and $y\in\B^m$,
Define $f_{\fcn}(xy)=1$ if $y_c=1$ for all $c\in f(x)$,
and define $f_{\fcn}(xy)=0$ if $y_c=0$ for all $c\in f(x)$;
if neither hold, then by definition
the string $xy$ is not in $\Dom(f_{\fcn})$.
\end{definition}

\begin{theorem}\label{thm:TFNPtoPartial}
Let $f$ be a query $\mathsf{TFNP}$ problem. Then
using $\tTheta$ and $\tOmega$ to hide $\polylog n$ factors,
\begin{enumerate}
\item $\D(f_{\fcn})=\tTheta(\D(f))$,
\item $\R(f_{\fcn})=\tTheta(\R(f))$,
\item $\Q(f_{\fcn})=O(\Q(f))$, $\Q(f_{\fcn})=\tOmega(\Q(f)^{1/3})$,
\item $\adeg(f_{\fcn})=O(\adeg(f))$, $\appdeg(f_{\fcn}) = O(\appdeg(f))$
%$\adeg(f_{\fcn})=\tOmega(\sqrt{\adeg(f)})$.
\item $\CAdv(f_{\fcn}) = \tOmega(\CAdv(f)), \Adv(f_{\fcn}) = \tOmega(\Adv(f))$.
\end{enumerate}
\end{theorem}

The theorem is proved via the following three lemmas.
\begin{lemma}
$\D(f_{\fcn}) = O(\D(f))$, $\R(f_{\fcn}) = O(\R(f))$, $\Q(f_{\fcn}) = O(\Q(f))$, $\adeg(f_{\fcn})=O(\adeg(f))$, $\appdeg(f_{\fcn}) = O(\appdeg(f))$.
\end{lemma}
\begin{proof}
The upper bounds for $\D(f_{\fcn})$, $R(f_{\fcn})$,
and $\Q(f_{\fcn})$ easy: on input $xy$, run the algorithm
for $f$ on $x$ to find a certificate $c$, then query $y_c$
and output the result. The upper bound for $\adeg(f_{\fcn})$
works similarly: simply construct the polynomial
$p(xy)=\sum_c p_c(x)y_c$. If the original polynomials were conical juntas, the resulting polynomial is also a conical junta; this gives the upper bound on $\appdeg(f_{\fcn})$.
\end{proof}

\begin{lemma}
$\D(f_{\fcn}) = \tOmega(\D(f))$, $\R(f_{\fcn}) = \tOmega(\R(f))$, $\Q(f_{\fcn}) = \tOmega(\Q(f))$.
\end{lemma}

\begin{proof}
To prove the lower bound on $\D(f_{\fcn})$, suppose we had
a deterministic algorithm $A$ for $f_{\fcn}$
which makes $T= \D(f_{\fcn})$ queries; we
construct an algorithm for $f$. On input $x$, run
$A(x0^m)$ (making true queries to $x$ as needed, and returning
$0$ whenever $A$ queries inside $y$).
This algorithm makes $T$ queries, and it must
query some bit $y_c$ for $c\in f(x)$, since it must change
its output if all such bits were flipped. We can then
check all certificates $c$ for the positions $y_c$ that
were queried by the run of $A$. There are at most $T$
such certificates, and each one uses $O(\polylog n)$
positions, so checking all such certificates takes only
$O(T\polylog n)$ queries. We output the first correct certificate
we find. This gives a valid deterministic algorithm for $\D(f)$,
completing the lower bound proof.

A similar argument works for lower bounding $\R(f_{\fcn})$.
A randomized algorithm $A$ for $f_{\fcn}$ which makes
error $\eps$ must, on input $x0^m$, query a bit
$y_c$ with $c\in f(x)$ with probability at least $1-2\eps$;
this is because it must detect whether the block
$B=\{c:c\in f(x)\}$ is flipped (and change its output
if so), and detecting such block flips requires querying
in the block via a standard argument. Hence, using the
same strategy of running $A(x0^m)$ and checking all certificates
$c$ for queried positions $y_c$, we find a certificate with
probability at least $1-2\eps$. We can repeat a constant
number of times to amplify the probability that we find a certificate.
This gives a valid algorithm for $f$.

For the quantum lower bound, we need the hybrid method \cite{BBBV97}.
In particular, the version in the appendix of \cite{BK19}
gives the following. Consider any quantum query algorithm $Q$,
and let $m_i^t$ be the probability that if we run $Q$ until
right before query number $t$ and then measure the query register,
the measurement outcome is $i$. Then clearly $\sum_i m_i^t=1$
for all $t$. Moreover, if $Q(x)$ accepts with probability
at most $\eps$ and $Q(y)$ accepts with probability at least
$1-\eps$, then the result says that
\[\sum_{t=1}^T\sum_{i:x_i\ne y_i} m_i^t
\ge\frac{(1-2\eps)^2}{4T},\]
where $T$ is the total number of queries made by $Q$.
Using this result, we can proceed as follows. Let $Q$
be any quantum algorithm for $f_{\fcn}$ making $T=\Q(f_{\fcn})$
queries, and let $\eps\le 1/3$ be its error.
Run $Q$ on $x0^m$, but measure its query register
right before its first query. Then restart it, and measure
the query register of $Q$ right before its second query.
Proceed this way $T$ times, to get $T$ different query
positions of $Q$ on $x0^m$. Since the set $\{c:c\in f(x)\}$
is a sensitive block of $x0^m$, the probability that
at none of these query positions we recovered corresponds
to a valid certificate is 
\[\prod_{t=1}^T \left(1-\sum_{c\in f(x)}m_c^t\right)
\le \prod_{t=1}^T\exp\left(-\sum_{c\in f(x)}m_c^t\right)
=\exp\left(-\sum_{t=1}^T\sum_{c\in f(x)}m_c^t\right)
\le \exp\left(-\frac{(1-2\eps)^2}{4T}\right).\]
Hence the probability that at least one of these
query positions corresponds to a valid certificate
is at least $\Omega(1/T)$. We can now check all these certificates
to get a quantum algorithm for $f$ that succeeds with probability
$\Omega(1/T)$, and then we can amplify it by repeating it $T$
times (using the fact that we only need to succeed once,
because we know when we found a correct certificate) to get
the success probability up to a constant. The total
number of queries used by this algorithm is $O(T^3)$
(we lost a factor of $T$ to amplification, and we lost
another factor of $T$ because we had to rerun the algorithm $T$
times and measure its query register each time in order
to detect where the quantum algorithm queried).
%
% To complete the proof, we must show the lower bound for
% $\adeg(f_{\fcn})$, which we delegate to a separate lemma.
\end{proof}

\begin{lemma}
$\CAdv(f_{\fcn}) = \tOmega(\CAdv(f)), \Adv(f_{\fcn}) = \tOmega(\Adv(f))$.
\end{lemma}

\begin{proof}
To show the lower bound for $\CAdv(f_{\fcn})$, we take the optimal solution for the optimization problem for $\CAdv(f_{\fcn})$ and use it to give a feasible solution to the problem for $\CAdv(f)$. The value of the objective function for $\CAdv(f)$ with this weight scheme will be $\tO(\CAdv(f_{\fcn})$, which shows the result for $\CAdv$.

Let $\{m_{xy}(i)\}_{i\in[n]}\cup\{m_{xy}(c)\}_{c\in S}$ be the optimal weight scheme for $\CAdv(f_{\fcn})$. For an input $x$ to $f$, let $B_x$ denote the block in the $y$ part of the input to $f_{\fcn}$ that corresponds to the certificates $c\in f(x)$. We'll use $y^{B_x}$ to refer to the string $y$ with the bits in $B_x$ flipped. The weight scheme for $\CAdv(f)$ is going to be the following:
\[ m'_x(i) = m_{x0^m}(i) + m_{x(0^m)^{B_x}}(i) + \sum_{c \in S: c_i \neq *}(m_{x0^m}(c) + m_{x(0^m)^{B_x}}(c)).\]
For $x, x'$ which have no certificates in common, consider the pair of $0$ and $1$-inputs $x0^m$ and $x'(0^m)^{B_{x'}}$ for $f_{\fcn}$. We have, 
\begin{align}
& \sum_{i\in[n]:x_i\neq x'_i}\min\left\{m_{x0^m}(i),m_{x'(0^m)^{B_{x'}}}(i)\right\} + \sum_{c\in S: (0^m)_c \neq ((0^m)^{B_{x'}})_c}\min\left\{m_{x0^m}(c),m_{x'(0^m)^{B_{x'}}}(c)\right\} \nonumber \\
& = \sum_{i\in[n]:x_i\neq x'_i}\min\left\{m_{x0^m}(i),m_{x'(0^m)^{B_{x'}}}(i)\right\} + \sum_{c\in f(x')}\min\left\{m_{x0^m}(c),m_{x'(0^m)^{B_{x'}}}(c)\right\} \nonumber \\
& \geq 1. \label{eq:fcn_weight}
\end{align}
Now if $x$ and $x'$ don't have any certificates in common, for every $c\in f(x')$, there must be some $i\in[n]$ such that $c_i\neq *$ and $x_i\neq x'_i$. Let's pick one such representative $i_c$ for each $c\in f(x')$ (there may be overlaps), and call the set of such $i_c$-s $S$. We have,
\begin{align*}
& \sum_{i\in[n]:x_i\neq x'_i}\min\{m'_x(i), m'_{x'}(i)\} \\
& \geq \sum_{i\notin S:x_i\neq x'_i}\min\left\{m_{x0^m}(i), m_{x'(0^m)^{B_{x'}}}(i)\right\} \\
& \quad + \sum_{i\in S}\min\left\{m_{x0^m}(i) + \sum_{c \in f(x') : i_c=i}m_{x0^m}(c), m_{x'(0^m)^{B_{x'}}}(i) + \sum_{c\in f(x'): i_c=i}m_{x'(0^m)^{B_{x'}}}(c)\right\} \\
& \geq \sum_{i\notin S:x_i\neq x'_i}\min\left\{m_{x0^m}(i), m_{x'(0^m)^{B_{x'}}}(i)\right\} + \sum_{i\in S}\min\left\{m_{x0^m}(i), m_{x'(0^m)^{B_{x'}}}(i) \right\}\\
& \quad + \sum_{i\in S}\sum_{c \in f(x') : i_c=i}\min\left\{m_{x0^m}(c), m_{x'(0^m)^{B_{x'}}}(c)\right\} \\
& = \sum_{i\in[n]:x_i\neq x'_i}\min\left\{m_{x0^m}(i),m_{x'(0^m)^{B_{x'}}}(i)\right\} + \sum_{c\in f(x')}\min\left\{m_{x0^m}(c),m_{x'(0^m)^{B_{x'}}}(c)\right\} \\
& \geq 1.
\end{align*}
The second inequality above is obtained by dropping some terms from the definitions of $m'_{x}(i)$ and $m'_{x'}(i)$; the third inequality is obtained by observing $\min\{a+b,c+d\} \geq \min\{a,c\} + \min\{b,d\}$; the fourth line simply reorganizes terms from the third line; the fifth line uses \eq{fcn_weight}. So clearly $m'_x(i)$ is a feasible weight scheme for $\CAdv(f)$. For any $x$, we have under this weight scheme,
\begin{align*}
\sum_{i\in[n]}m'_x(i) & = \sum_{i\in[n]}(m_{x0^m}(i) + m_{x(0^m)^{B_x}}(i)) + \sum_{i\in[n]}\sum_{c\in S:c_i\neq *}(m_{x0^m}(c) + m_{x(0^m)^{B_x}}(c)) \\
 & = \sum_{i\in[n]}(m_{x0^m}(i) + m_{x(0^m)^{B_x}}(i)) + \sum_{c\in S}|c|(m_{x0^m}(c) + m_{x(0^m)^{B_x}}(c)) \\
 & \leq \C(f)\left(\sum_{i\in[n]}m_{x0^m}(i) + \sum_{c\in S}m_{x0^m}(c)\right) + \C(f)\left(\sum_{i\in[n]}m_{x(0^m)^{B_x}}(i) + \sum_{c\in S}m_{x(0^m)^{B_x}}(c)\right) \\
 & \leq 2\C(f)\CAdv(f_{\fcn}).
\end{align*}
In first inequality above, we used the fact that $|c| \leq \C(f)$ for all $c\in S$; in the second inequality, we used that for any $xy$, $\sum_im_{xy}(i) + \sum_cm_{xy}(c) \leq \CAdv(f_{\fcn})$. Noting that $\C(f) = \polylog(n)$ completes the proof.

The proof for $\Adv$ works very similarly: the assignment of the weight scheme for $\Adv(f)$ is done in exactly the same way using the weight scheme for $\Adv(f_{\fcn})$. In order to show feasibility, we'll need to use $\sqrt{(a+b)(c+d)} \geq \sqrt{ac}+\sqrt{bd}$, where we previously used $\min\{a+b,c+d\} \geq \min\{a,c\} + \min\{b,d\}$.
\end{proof}

\section{Polynomials, conical juntas and quantum query complexity}

In this section, we show the separation between $\appdeg$ and $\Q$ for a search problem, which will lead to a separation for a partial function by \thm{TFNPtoPartial}.
\subsection{Error reduction for polynomials}
In order to show our separation, we first need to show that when $f$ is a search problem, we can amplify polynomials to reduce the error.

\begin{lemma}\label{lem:error-red}
For any search problem $f$, and $\eps, \delta \in (0,1)$,
\[ \udeg_\delta(f) = O\left(\frac{\log(1/\delta)}{\eps}(\udeg_{1-\eps}(f)+\C(f))\right).\]
Moreover, if the polynomial achieving $\udeg_{1-\eps}(f)$ satisfies $\sum_{c\in S} p_c(x) \leq b$, for some $b$ potentially bigger than 1, on $x\notin \Dom(f)$, then the polynomial achieving $\udeg_\delta(f)$ satisfies $\sum_{c\in S} q_c(x) \leq b^{O(\log(1/\delta)/\eps)}$ on $x\notin \Dom(f)$.
\end{lemma}
\begin{proof}
Let $d:= \adeg_{1-\eps}(f)$ and let $\{p_c\}_{c\in S\cup\{\bot\}}$ be the collection of polynomials achieving this. For $k$ to be determined later, let $\bar{c}$ denote the $k$-tuple $(c^1,\ldots, c^k)$ where each $c^i\in S\cup\bot$. We define the polynomials 
\[p_{\bar{c}}(x_1,\ldots, x_n) = \prod_{c^i\in \bar{c}}p_{c^i}(x_1,\ldots, x_n).\]
The degree of $\{p_{\bar{c}}\}$ is $kd$, and intuitively it represents the probability distribution of for seeing output tuple $\bar{c}$ on repeating the polynomial system $\{p_c\}$ independently $k$ times. For each $c\in S\cup\{\bot\}$, let $c(x)$ be the polynomial which checks if an input $x$ satisfies the certificate $c$. It can be seen that the conjunction $\prod_{c_i=1}x_i\prod_{c_i=0}(1-x_i)$ satisfies this requirement for $c\in S$. For $c=\bot$, $c(x)$ is identically $0$. Hence the degree of any $c(x)$ is at most $\C(f)$.

Using the two sets of polynomials $p_{\bar{c}}(x)$ and $c(x)$, we define the following collection of approximating polynomials for $f$:
\begin{equation}\label{eq:error-red}
q_{c}(x) = \begin{cases} \displaystyle\sum_{i=1}^k\sum_{\bar{c}: c^i=c}\prod_{j<i}(1-c^j(x))c(x)p_{\bar{c}}(x) & \text{ if } s\in S \\
\displaystyle\sum_{\bar{c}}\prod_{i=1}^k(1-c^i(x))p_{\bar{c}}(x) & \text{ if } c=\bot. \end{cases}
\end{equation}
$\prod_{j<i}(1-c^j(x))c(x)p_{\bar{c}}(x)$ is the indicator variable for the event that $c^i=c$ is the first certificate corresponding to $\bar{c}$ that is satisfied on $x$. $\prod_{i=1}^k(1-c^i(x))$ is the indicator variable for no certificate corresponding to $\bar{c}$ being satisfied on $x$.
Therefore, intuitively, the polynomials $\{q_c\}$ represent the following procedure:
\begin{enumerate}
\item Run $\{p_c\}$ $k$ times to get outcome $\bar{c}$
\item For $i=1$ to $k$, check if $c^i$ is valid on $x$ (which implies $c^i$ is a valid output)
\begin{enumerate}[label=(\Alph*)]
\item If valid, output $c^i$
\item If invalid, move on to the next $i$
\end{enumerate}
\item If no $c^i$ is valid on $x$, output $\bot$.
\end{enumerate}

In order to prove that $\{q_c\}$ approximate $f$, we use the probabilistic interpretation of the polynomials on inputs $x\in \Dom(f)$. Firstly, it is easy to see that $q_c(x)$ is positive for every $x\in \Dom(f)$ because every term in it is positive. For $c\neq \bot$, we have,
\begin{align*}
q_c(x) & = \sum_{i=1}^k\sum_{\bar{c}:{c^i}=c}\Id[c^i \text{ is the first certificate in $\bar{c}$ satisfied on } x]\cdot\Pr[\text{output $\bar{c}$ on } x] \\
 & = \sum_{i=1}^k\sum_{\bar{c}}\Id[c^i=c]\Id[c^i \text{ is the first certificate in $\bar{c}$ satisfied on } x]\cdot\Pr[\text{output $\bar{c}$ on } x] \\
\end{align*}
For $c=\bot$ we have,
\begin{align*}
q_\bot(x) & = \sum_{\bar{c}}\Id[\text{no certificate in $\bar{c}$ is satisfied on } x]\cdot\Pr[\text{output $\bar{c}$ on } x] \\
\end{align*}
Therefore we have for any $x\in \Dom(f)$,
\begin{align*}
\sum_{c\in S\cup \{\bot\}}q_c(x) & = \sum_{i=1}^k\sum_{\bar{c}}\sum_c\Id[c^i=c]\Id[c^i \text{ is $1$st certificate in $\bar{c}$ satisfied on } x]\cdot\Pr[\text{output $\bar{c}$ on } x] \\
 & \quad + \sum_{\bar{c}}\Id[\text{no certificate in $\bar{c}$ is satisfied on } x]\cdot\Pr[\text{output $\bar{c}$ on } x] \\
 & = \sum_{i=1}^k\sum_{\bar{c}}\Id[c^i \text{ is the first certificate in $\bar{c}$ satisfied on } x]\cdot\Pr[\text{output $\bar{c}$ on } x] \\
 & \quad + \sum_{\bar{c}}\Id[\text{no certificate in $\bar{c}$ is satisfied on } x]\cdot\Pr[\text{output $\bar{c}$ on } x] \\
 & = \sum_{\bar{c}}\Id[\text{some certificate in $\bar{c}$ is satisfied on } x]\cdot\Pr[\text{output $\bar{c}$ on } x] \\
 & \quad + \sum_{\bar{c}}\Id[\text{no certificate in $\bar{c}$ is satisfied on } x]\cdot\Pr[\text{output $\bar{c}$ on }x] \\
 &  = \sum_{\bar{c}}\Pr[\text{output $\bar{c}$ on } x] =1.
\end{align*}
Thus $\{q_c\}$ satisfies the normalization condition for approximating polynomials. To see what the approximating factor for the set is, We have to upper bound $q_\bot(x)$, as the other polynomials $q_c$ are by definition only non-zero when $c$ is a valid output on $x$.
\begin{align*}
q_\bot(x) & = \sum_{\bar{c}}\Id[\text{no certificate in $\bar{c}$ is satisfied on } x]\cdot\Pr[\text{output $\bar{c}$ on } x] \\
 & = \Pr[\text{no output is valid on $x$ in $k$ independent runs of } \{p_c\}] \\
 & = \Pr[\text{no output is valid on $x$ in }\{p_c\}]^k \leq (1-\eps)^k.
\end{align*}
Taking $k= \frac{\log(1/\delta)}{\log(1/(1-\eps))} = O\left(\frac{\log(1/\delta)}{\eps}\right)$ (for small $\eps$), the above quantity is at most $\delta$. The degree of $\{q_c\}$ is $k\cdot\deg(\{c\}) + \deg(\{p_{\bar{c}}\}) = k(\C(f)+d) = O\left(\frac{\log(1/\delta)}{\eps}(\C(f)+d)\right)$, which is the required upper bound.

To upper bound $\sum_cq_c(x)$ for $x\notin \Dom(f)$, we notice that for these inputs we still have,
\begin{align*}
\sum_{c\in S\cup\{\bot\}}q_c(x) & = \sum_{\bar{c}}\Id[\text{some certificate in $\bar{c}$ is satisfied on } x]\cdot p_{\bar{c}}(x) \\
 & \quad + \sum_{\bar{c}}\Id[\text{no certificate in $\bar{c}$ is satisfied on } x]\cdot p_{\bar{c}}(x) \\
  & = \sum_{\bar{c}}p_{\bar{c}}(x),
\end{align*}
though $p_{\bar{c}}(x)$ can no longer be interpreted as a probability. However, since $p_{c^i}(x) \geq 0$ for all $c^i\in\bar{c}$ and $x\notin \Dom(f)$,
\begin{align*}
\sum_{\bar{c}}p_{\bar{c}}(x) & = \sum_{c^1\ldots c^k}p_{c^1}\cdot\ldots\cdot p_{c^k}(x) \\
 & = \sum_{c^1}p_{c^1}(x)\ldots\sum_{c^k}p_{c^k}(x) \\
 & = \left(\sum_cp_c(x)\right)^k \leq b^k.
\end{align*}
Using the bound on $k$, we get the required upper bound for the sum of the system on $x\notin \Dom(f)$.
\end{proof}

It is easy to see in the above proof that if the polynomial was originally bounded, it remains bounded. With some modification, we can also make the polynomials $q_c$ conical juntas if the polynomials $p_c$ were conical juntas. In the polynomials $q_c$, the polynomials $p_c$ are multiplied with themselves, and with $c(x)$ and $1-c(x)$. All of these except $1-c(x)$ is a conical junta. Recall that $1-c(x)$ is the polynomial that checks that the certificate $c$ is \emph{not} satisfied on $x$. Let $P_c$ be the collection of partial assignments that only fix the variables that $c$ fixes, and fix them to every possible substring than the one fixed by $c$. The polynomial $\sum_{p\in P_c}p(x)$ has the same degree as $1-c(x)$, and the same effect. When $c$ is satisfied on $x$, it takes value 0; when $c$ is not satisfied, it takes value 1, because exactly one of the partial assignments in $P_c$ turns on. Therefore, we can get a collection of conical juntas $\{q_c\}$ with reduced error, in exactly the same way as the proof of \lem{error-red}, except replacing $1-c(x)$ with $\sum_{p\in P_c}p(x)$ for every $c$. This gives the following corollary.
\begin{corollary}
For any search problem $f$, and $\eps, \delta \in (0,1)$, the following hold:
\begin{align*}
\adeg_\delta(f) & = O\left(\frac{\log(1/\delta)}{\eps}(\adeg_{1-\eps}(f)+\C(f))\right) \\
\adeg\!{}_\delta^{+2}(f) & = O\left(\frac{\log(1/\delta)}{\eps}(\adeg\!{}_{1-\eps}^{+2}(f)+\C(f))\right).
\end{align*}
\end{corollary}

It was shown in \cite{GLM+16} that unlike approximate degree, it is not possible to do error reduction for conical junta degree in general (the standard error amplification procedure for boolean functions, involving repeating the procedure several times and taking majority, does not preserve the non-negativity of the polynomial). We are able to circumvent this barrier because our polynomials are required to output certificates, and we do a non-standard error reduction technique by checking certificates.

\subsection{Approximate degree upper bound}
\begin{lemma}\label{lem:Chernoff}
Let $n\ge 2$.
If $n^2$ i.i.d.\ bits $X_i$ are generated from a Bernoulli
distribution with probability of $1$ equal to $2/n$, then
\[\Pr\left[\sum_{i=1}^{n^2} X_i\notin [n,3n]\right]
\le 2^{-n/4}.\]
\end{lemma}

\begin{proof}
Let $X=\sum_i X_i$, and let $\mu=\E[X]=2n$.
A standard Chernoff bound gives
\[\Pr[X\ge (1+\delta)\mu]\le
\left(\frac{e^\delta}{(1+\delta)^{1+\delta}}\right)^\mu,\]
\[\Pr[X\le (1-\delta)\mu]\le
\left(\frac{e^{-\delta}}{(1-\delta)^{1-\delta}}\right)^\mu.\]
Using $\delta=1/2$, we get
\[\Pr[X\notin[n,3n]]\le (8e/27)^n+(2/e)^n.\]
This expression is asymptotically smaller than $2^{-n/4}$,
and one can compute that it starts being smaller from $n\ge 9$.
Finally, one can exactly compute the probability that
$X$ lies outside $[n,3n]$ for the cases $n=2,3,4,5,6,7,8$,
and observe that this is smaller than $2^{-n/4}$ for those cases.
\end{proof}

\begin{definition}
Let $n\ge 4$ be a power of $2$,
and let $\AndOr_n$ be the Boolean function
on $n^2\log_2(n/2)$ bits defined as an OR of size $n^2$
composed with an AND of size $\log_2(n/2)$; that is,
the input is divided into $n^2$ blocks of $\log_2(n/2)$ bits each,
and the function returns $1$ if there is an all-$1$ block
and returns $0$ otherwise.

The search problem corresponding to this function will also be denoted by $\AndOr_n$, and will refer to
the task of finding a $1$-certificate when given a $1$-input
(i.e.\ outputting the location of an all-$1$ block,
promised that such a block exists).
\end{definition}

\begin{theorem}\label{thm:1poly}
There is a collection of unbounded polynomials of degree $\log(n/2)$ which
solves $\AndOr_n$ with success probability at least $1/3$
on all but a $2^{-n/4}$ fraction
of the inputs. Moreover, on all inputs, the sum of the system
is in $[0,n/3]$.
\end{theorem}

\begin{proof}
For $i\in[n^2]$, $j\in[\log(n/2)]$, and $x\in\B^{n^2\log(n/2)}$,
let $x_{ij}$ denote bit $j$ in block $i$; that is, the index $i$
indexes an AND gate, and the index $j$ indexes one of the $\log(n/2)$
bits in the AND gate. Note that each potential certificate is
indexed by $i$: the certificate corresponding to block $i$ requires $x_{ij}=1$ for all $j\in [\log(n/2)]$. We'll just use $i$ to refer to this certificate.

Define the polynomials
\[p_i(x)=\frac{1}{3n}\prod_{j\in[\log(n/2)]} x_{ij}.\]
Then each $p_i$ has degree $\log(n/2)$. Say an input $x$
is ``good'' if the number of AND gates evaluating to $1$
in $x$ is in $[n,3n]$. Then \lem{Chernoff} tells us the fraction
of bad inputs is at most $2^{-n/4}$. On a good input $x$,
$\sum_{i\in[n^2]} p_i(x)\in[1/3,1]$. Therefore, the system
of polynomials $\{p_i\}_i$ solves $f_n$ with success probability
at least $1/3$ on all but a $2^{-n/4}$ fraction of the inputs.
On the bad inputs, the sum $\sum_i p_i(x)$ is still non-negative
and bounded above by $n/3$, as desired.
\end{proof}

% \begin{theorem}[Amplification]
% If a system of polynomials solves a search problem to error $\eps$,
% there is another system whose degree is at most $k$ times larger
% which solves the search problem to error $\eps^k$. Moreover,
% whenever the sum of the original system is in $[0,1]$ on some
% input, this remains true for the new system, and whenever
% the sum of the original system is at most $T$ for some input,
% the new system's sum on that input is at most $(4T)^k$.
% \end{theorem}

% \begin{proof}
% \comment{TODO}
% \end{proof}
The following corollary is obtained by applying \lem{error-red} to \thm{1poly}, noticing that the certificates of $f_n$ are of size $\log(n/2)$.
\begin{corollary}\label{cor:1poly-amp}
There is an unbounded system of polynomials of degree $O(\log^2 n)$
which solves $\AndOr_n$ to error $1/\sqrt{n}$ on all but a
$2^{-n/4}$ fraction of the inputs. Moreover, on all inputs,
the sum of the system is in $[0,2^{O(\log^2 n)}]$.
\end{corollary}

%\comment{TODO: add citation for this}
We now employ the following folklore result.

\begin{lemma}[Existence of $k$-wise independent functions]\label{lem:kwise}
There exists a set $\clH$ of $2^{mk}$ functions from $\{0,1\}^m$ to $\{0,1\}$, which are $k$-wise independent. That is, for distinct $x_1, \ldots x_k \in \{0,1\}^m$ and any $y_1,\ldots, y_k \in \{0,1\}$,
\[ \Pr_{h\in \clH}[h(x_1)=y_1, \ldots, h(x_k)=y_k] = 2^{-mk}, \]
where the probability is over uniform $h \in \clH$.
\end{lemma}

The definition of $k$-wise independence implies that  if $\clH$ is a set of $k$-wise independent functions from $[n^2\log(n/2)]$ to $\B$, then for a random $h$ from $\clH$, $(h((i,j)))_{(i,j)\in [n^2\log(n/2)]}$ looks indistinguishable to from a uniformly random string in $\{0,1\}^{n^2\log(n/2)}$ to a degree $k$ polynomial.

\begin{definition}
Let $n \geq 4$ be a power of 2, and let $\clH : [n^2\log(n/2)]\to \B$ be a family of $\log^2 n$-wise independent functions from \lem{kwise}. For $x\in \B^{n^2\log(n/2)}$, let $x\oplus h$ refer to the string in $\B^{n^2\log(n/2)}$ whose $(i,j)$-th bit is given by $x_{ij}\oplus h((i,j))$. For $n$ implicit, and any integer $t$, define the search problem $\AndOr_{\clH,t} : \B^{tn^2\log(n/2)}$ as follows: for $x^1\ldots x^t$, where each $x^\ell \in \B^{n^2\log(n/2)}$, the output certificates for $\AndOr_{\clH,t}$ on $x^1\ldots x^t$ consist of
\[ \left\{(i^1, \ldots, i^t) : \exists h \in \clH \text{ s.t. } i^1\in \AndOr_n(x^1\oplus h), \ldots, i^t \in \AndOr_n(x^t\oplus h) \right\}.\]
\end{definition}
Note that $\AndOr_n$ is not a total search function. In the above definition, $h$ is supposed to such that $x^\oplus h, \ldots, x^t\oplus h$ are all inputs to $\AndOr_n$ which do have certificates. It is not clear that such an $h$ exists for all inputs $x^1\ldots x^t$, so that $\AndOr_{\clH,t}$ is total. We shall argue that this is indeed the case for $t=\polylog(n)$ in the following theorem.

\begin{theorem}\label{thm:adeg-ub}
For $t = \polylog(n)$, $\AndOr_{\clH,t}$ is a total function on $\B^{tn^2\log(n/2)}$, and its approximate degree satisfies $\adeg(\AndOr_{\clH,t}) = O(t\log^2 n)$.
\end{theorem}
\begin{proof}
For $t= \polylog(n)$, we shall show a system of bounded polynomials of degree $O(t\log^2 n)$ that finds a valid output with probability $\frac{2}{3}$ for any input $x^1\ldots x^t$ to $\AndOr_{\clH,t}$; this implies that $\AndOr_{\clH,t}$ must be total.

Let $\{p_i\}_i$ be the unbounded polynomials from \cor{1poly-amp} that compute the search problem $f_n$  with error $\frac{1}{\sqrt{n}}$ for $1-2^{-n/4}$ fraction of inputs in $\B^{n^2\log(n/2)}$. Define the system of polynomials $q_{h,i^1,\ldots, i^t}$ for $\AndOr_{\clH,t}$ by
\[ q_{h,i^1,\ldots, i^t}(x^1\ldots x^t) = \frac{1}{(1+\eps)|\clH|}\cdot p_{i^1}(x^1\oplus h)\cdot\ldots \cdot p_{i^t}(x^t\oplus h),\]
for an $\eps=t\cdot 2^{O(t\log^2 n)-n/4}$. Since the degree of each $p_{i^\ell}$ is $O(\log^2n)$, it is clear that the degree of the set of polynomials $q$ is $O(t\log^2 n)$. The actual outputs of $\AndOr_{\clH,t}$ are indexed only by $i^1,\ldots, i^t$, so we'll get the final polynomials by summing $q_{h,i^1,\ldots, i^t}$ over $h$. 

First we shall show that the polynomials $\{q_{h,i^1,\ldots, i^t}\}$ satisfy the normalization condition. Fix an input $x^1\ldots x^t$, and for $\ell \in [t]$, let $\Bad^\ell$ refer to the set of $h$-s for which $x^\ell\oplus h$ is in the $2^{-n/4}$ fraction of inputs on which $\{p_i\}$ does not work. If $h$ were a truly random function, then $x^\ell\oplus h$ is a truly random string in $\B^{n^2\log(n/2)}$, and we know $\Pr_h[\Bad^\ell] \leq 2^{-n/4}$ in this case. However, the event $\Bad^\ell$ is defined by degree $\log^2 n$ polynomials, which cannot differentiate truly random functions from $\log^2 n$-wise independent functions. Therefore, even when $h$ is picked uniformly from the set $\clH$ instead, we must have $\Pr_h[\Bad^\ell] \leq 2^{-n/4}$ for all $\ell$. By the union bound, $\Pr_h[\Bad^1\cup\ldots \cup\Bad^t] \leq t\cdot 2^{-n/4}$. We know that when $h\in (\Bad^1\cup\ldots\cup\Bad^t)^c$, $\sum_{i^\ell}p_{i^\ell}(x^\ell\oplus h) \leq 1$ for each $\ell$. When $h\in \Bad^1\cup\ldots\Bad^t$, $\sum_{i^\ell}p_{i^\ell}(x^\ell\oplus h) \leq 2^{O(\log^2 n)}$ for every $\ell$. Therefore, for any $x^1\ldots x^t$,
\begin{align*}
\sum_{h,i^1,\ldots,i^t}q_{h,i^1,\ldots i^t}(x^1\ldots x^t)  & = \frac{1}{1+\eps}\left(\sum_{h\in (\Bad^1\cup\ldots\cup\Bad^t)^c}\frac{1}{|\clH|}\sum_{i^1\ldots i^t}p_{i^1}(x^1\oplus h)\cdot\ldots \cdot p_{i^t}(x^t\oplus h) \right. \\
 & \quad \left. + \sum_{h\in \Bad^1\cup\ldots\cup\Bad^t}\frac{1}{|\clH|}\sum_{i^1\ldots i^t}p_{i^1}(x^1\oplus h)\cdot\ldots \cdot p_{i^t}(x^t\oplus h)\right) \\
 & \leq \frac{1}{1+\eps}\left(\sum_{h\in (\Bad^1\cup\ldots\cup\Bad^t)^c}\frac{1}{|\clH|}\cdot 1^t + \sum_{h\in \Bad^1\cup\ldots\cup\Bad^t}\frac{1}{|\clH|}\cdot 2^{O(\log^2 n)\cdot t}\right) \\
 & = \frac{\Pr_h[(\Bad^1\cup\ldots\cup\Bad^t)^c] + \Pr[\Bad^1\cup\ldots\cup\Bad^t]\cdot 2^{O(\log^2n)\cdot t}}{1+\eps} \\
 & \leq \frac{1 + (2^{O(t\log^2 n)}-1)\cdot 2^{-n/4}}{1+\eps} \leq 1,
\end{align*}
by the choice of $\eps$.

To show that the system of polynomials succeed with probability at least $\frac{2}{3}$, we note that when $h\in(\Bad^1\cup\ldots\cup\Bad^t)^c$, then for each $\ell$, $p_{i^\ell}$ fails on $x^\ell\oplus h$ with probability at most $\frac{1}{\sqrt{n}}$. Therefore, when $h\in (\Bad^1\cup\ldots\cup\Bad^t)^c$, the probability that any $p_{i^\ell}$ fails on any $x^\ell\oplus h$ is at most $\frac{t}{\sqrt{n}}$. This gives us,
\begin{align*}
\sum_{(h,i^1,\ldots,i^t) \in \AndOr_{\clH,t}(x^1\ldots x^t)}q_{h,i^1,\ldots i^t}(x^1\ldots x^t) & \geq \sum_{h\in(\Bad^1\cup\ldots\cup\Bad^t)^c}\sum_{i^1\ldots i^t: i^\ell\in f_n(x^\ell\oplus h) \, \forall \ell}q_{h,i^1,\ldots i^t}(x^1\ldots x^t) \\
 & \geq \frac{1}{1+\eps}\sum_{h\in (\Bad^1\cup\ldots\cup\Bad^t)^c}\frac{1}{|\clH|}\cdot\left(1-\frac{t}{\sqrt{n}}\right) \\
 & \geq \frac{(1-2^{-n/4})(1-t/\sqrt{n})}{1+\eps} \geq \frac{2}{3},
\end{align*}
by the choice of $t$. This completes the proof.
\end{proof}

The following corollary is obtained by noticing that the polynomials $p_i$ in \thm{1poly} were conical juntas, and therefore, the polynomials $q_{h,i^1,\ldots, i^t}$ in \thm{adeg-ub} are conical juntas as well.
\begin{corollary}\label{cor:jdeg-ub}
For $t=\polylog(n)$, $\appdeg(\AndOr_{\clH,t})  = O(t\log^2n)$.
\end{corollary}

\subsection{Quantum lower bound}

We work with degree of polynomial systems, since this
will be useful later in the context of query-to-communication lifting.
We will use the following lemma in order to talk about
polynomial systems in an algorithmic way.

\begin{lemma}
Systems of polynomials can implement the following algorithm-like
operations:
\begin{enumerate}
\item Making choices randomly: if we have a probability
distribution over systems of polynomials of degree at most $d$,
then there is a system of polynomials with degree at most $d$
for which the probability of outputting a symbol $c$ on an
input $x$ is the same as that of sampling a system of polynomials
from the distribution, running it on $x$, and getting the output $c$.
\item Replacing, duplicating, flipping, or permuting bits:
if we have a system of polynomials $p_c$ acting on strings
of length $n$ and an input string $x$ of length $m$,
then for any procedure $\varphi$ which maps strings of length $m$
to strings of length $n$ by duplicating or deleting variables,
adding constant bits, flipping variables, or permuting the variables,
there is a system of polynomials $p'_c$ which has the same
degree as the system $p_c$ and satisfies $p'_c(x)=p_c(\varphi(x))$
for all $x$ and $c$.
\item Checking certificates: if $p_c$ is a system of polynomials
of degree at most $d$ and if each certificate $c$ reveals at most
$k$ bits of the input, then there is a system of polynomials
$p'_c$ of degree at most $d+k$ such that $p'_c(x)=0$ if
$c$ is not in $x$ and $p'_c(x)=p_c(x)$ if $c$ is in $x$.
\item Adaptive computation: suppose we have a decision tree
specified as follows. At each intermediate node, there
is a system of polynomials with some type of output symbol;
the outgoing arcs from that node are labeled by the output
symbol for that system, corresponding to a deterministic
algorithm which makes decisions based on the output
of the system. The leaves of the decision tree are labeled by
yet another type of output symbol.
Then for any such decision tree, there
is a system of polynomials simulating it, such that
the degree of the system is at most the maximum sum of degrees
along any root-to-leaf path of the decision tree, and such that
the weight of each output symbol (on any input $x$)
is the ``probability of reaching that symbol'': that is, it
is the same as the sum of all products of weights along
root-to-leaf paths in the decision tree that lead to leaves
labeled by that symbol.
\end{enumerate}
\end{lemma}

\begin{proof}
These constructions are mostly straightforward.
For the first, if we have a probability distribution over systems
of polynomials, we can average the polynomials $p_c$
for each $c$ separately, where the average is weighted by
the probability distribution. It is easy to check that this
has the desired properties: the largest degree cannot
increase, and the value of the new, averaged $p_c$ (evaluated on
input $x$) is the expected value of the old $p_c(x)$ values.

For the second, we can plug in duplicates of the variables
$x_i$, or constants, into the inputs of the polynomials
$p_c$ in the correct permutation. This does not increase
the degree of the polynomial system.

For checking certificates, we can multiply each
polynomial $p_c$ in the system by a polynomial for $c$,
which on input $x$ evaluates to $1$ if $c$ is in $x$
and to $0$ otherwise; this checking polynomial has degree
$|c|\le k$, and the new system therefore has degree at most
$d+k$.

Finally, we discuss adaptive computation. For each output
symbol of the decision tree, we construct a polynomial
in the simulating system as follows: we take the sum
over all root-to-leaf paths in the tree for paths
labeled by that symbol, and for each such path,
we take the product of all polynomials in the intermediate
nodes along that path which correspond to the arcs of that path.
The conditions on the degree and output weight are easy to check.
\end{proof}

We now need the following direct product result by Sherstov.

\begin{theorem}[\cite{She12}]\label{thm:query-dp}
There is a universal constant $\beta>0$ such that the following holds.
For a partial boolean function $g$, let $g^n$ denote the relational problem corresponding to solving $n$ copies of $g$.
%Then there exists a constant $\beta > 0$ such that
Then for all (possibly partial) $g$ and all $n\in\bN$,
\[ \adeg_{1-2^{-\beta n}}(g^n) \ge \beta n\cdot\adeg(g).\]
\end{theorem}

We use this to prove a lower bound on finding all the ones
in a string, even with low success probability.

\begin{theorem}\label{thm:DegGroverCopies}
Let $f_{n,k}$ be the following relational problem: given a string of length $n$
and Hamming weight $k$, find all $k$ $1$s in the string. Then there is a universal
constant $\alpha>0$ such that for all $n\in\bN$ and all $k\le n/2$,
\[\adeg_{1-2^{-\alpha k}}(f_{n,k}) \ge \alpha \sqrt{nk}.\]
\end{theorem}

% \begin{proof}
% From \thm{query-dp}, we know that there is a universal constant $\beta>0$
% such that
% \[\adeg_{1-2^{-\beta \ell}}(\PromiseOR_m^\ell)\ge \beta \ell\sqrt{m};\]
% this is because $\adeg(\PromiseOR_m)\ge\Omega(\sqrt{m})$ \cite{BBC+01}.

% \end{proof}

% \begin{proof}
% From \thm{query-dp}, we know that there is a universal constant $\beta>0$
% such that
% \[\adeg_{1-2^{-\beta \ell}}(\PromiseOR_m^\ell)\ge \beta \ell\sqrt{m};\]
% this is because $\adeg(\PromiseOR_m)\ge\Omega(\sqrt{m})$ \cite{BBC+01}.
% Apply \lem{FindOnesRepeatedly} with $\alpha=\beta/2$ to get the
% constant $C$. 

% \end{proof}

\begin{proof}
From \thm{query-dp}, we know that there is a universal constant $\beta>0$
such that
\[\adeg_{1-2^{-\beta \ell}}(\PromiseOR_m^\ell)\ge \beta \ell\sqrt{m};\]
this is because $\adeg(\PromiseOR_m)\ge\Omega(\sqrt{m})$ \cite{BBC+01}.
Recall that $\binom{N}{\alpha N}\le 2^{H(\alpha)N}$ where
$H(\cdot)$ is the binary entropy function. Note $H$ is a bijection
from $[0,1/2]$ to $[0,1]$.
%Let $\alpha=H^{-1}(\beta/2)/2\le 1/4$.
Let $c=\lceil 2/H^{-1}(\beta/2)\rceil\ge 4$, which is an absolute constant.
%Then for any $N\ge 2/\alpha$ and any $K\le \lceil \alpha N\rceil$,
%we have $\binom{N}{K}\le 2^{\beta N/2}$.
Then for any $N\ge 2c$ and any $K\le N/c$, we have
$\binom{N}{K}\le 2^{\beta N/2}$.

We now show how to use a system of polynomials solving $f_{n,k}$
with not-too-tiny success probability in order to solve
$\PromiseOR_m^{\ell}$ with not-too-tiny success probability.
Assume for now that $n\ge 2ck$.
We pick $\ell=ck$ and $m=\lfloor n/\ell\rfloor\ge 2$.
%We pick $\ell=\lceil k/\alpha\rceil$ and $m=\lceil \alpha n/k\rceil$.
% Note that $k/\alpha\ge 4$, so $\ell \le (5/4)(k/\alpha)$;
% assuming $k\le \alpha n$, we also have $m\le 2\alpha n/k$,
% so $n\le \ell m\le 3n$.

% To describe our reduction, we informally
% talk about polynomial systems as if they are algorithms.
We wish to find
all the ones (up to $\ell$ of them) in an input to $\PromiseOR_m^\ell$,
which has $\ell m\in[n/2,n]$ bits.
We have a polynomial system $P$ that can find $k$ ones in an input of size
$n$ which has exactly $k$ ones. On an input of size $m\ell$, we can pad
it with up to $n/2$ zeroes to get an input of size $n$.
The input is then of size $n$, but the number of ones is only known
to be between $0$ and $\ell=ck$, rather than being exactly $k$.
We handle the number of ones being less than $k$ separately from the
case where it is at least $k$ (we can guess which of these is the case,
succeeding with probability $1/2$).

When the number of ones is at least $k$,
the next step is to guess the number of ones by picking
an integer uniformly in $k,k+1,\dots, ck$. Given the guess $g$,
we then randomly replace some of the $n$ input bits with $0$,
leaving each bit alive with probability $k/g$. If there really were
$g$ ones in the input, then on expectation there are now $k$ ones,
and it is not hard to show that there are exactly $k$ ones with
probability at least $\Omega(1/\sqrt{k})$, for any fixed $k$ and $g$.
We can now plug in the resulting string into the system of polynomials,
which finds the $k$ ones with some probability $\delta$ if there are exactly $k$
ones. Assuming the guess $g$ was correct, this process finds $k$ random
ones in the string with probability at least $\Omega(\delta/\sqrt{k})$.
If we repeat this process $g/k\le c$ times (zeroing out a different set of bits
each time), we can find a constant fraction of the ones with probability
at least $\Omega((\delta/\sqrt{k})^{c})$. If instead of $g/k$ times,
we repeat it $(g/k)\cdot O(\log c)$ times, we will reduce the fraction
of unfound ones to $1/c$, leaving fewer than $k$ unfound ones;
this succeeds with probability $(\delta/\sqrt{k})^{O(c\log c)}$.
This leads to a system of polynomials for finding all but at most $k$ ones
with the aforementioned probability and with degree at most $O(d\cdot c\log c)$,
where $d$ is the degree of the original system $P$.

We have now reduced to the case where there are fewer than $k$
ones remaining to be found in the string. If we guess the number
of ones remaining
and then guess their positions and check that they
are all $1$, we can find all the ones with probability at least
$\frac{1}{k}\binom{n}{k-1}^{-1}\ge 2^{-\beta n/2}/k$ and degree at most $k-1$.
Combining this with the polynomial system from the previous paragraph,
we can find all the ones in the string with degree $O(d+k)$ and probability
at least $\left(\frac{\delta}{k}\right)^{O(1)}2^{-\beta n/2}$.
It follows that either this degree is at least $\beta\ell\sqrt{m}=\Omega(\sqrt{nk})$,
which means $d=\Omega(\sqrt{nk})$, or else this success probability
is less than $2^{-\beta\ell}$, which means 
$\left(\frac{\delta}{k}\right)^{O(1)}<2^{-\beta\ell/2}$
or $\log (k/\delta)=\Omega(\beta\ell/2)=\Omega(k)$.
The latter implies $\log(1/\delta)=\Omega(k)$, or $\delta\le 2^{-\Omega(k)}$.

This proves the desired result under the assumption that $n\ge 2ck$.
It remains to handle the case of $k\in[n/2c,n/2]$. We use another reduction:
we show that a polynomial system solving $f_{n,k}$ for such $k$ can
be used to solve $f_{n,n/2c}$. Given a system $P$ solving $f_{n,k}$
as well as an input $x$ of Hamming weight $n/2c$, we randomly replace
some bits of $x$ with $1$s such that the number of $1$s becomes $k$
with probability at least $\Omega(1/\sqrt{n})$. We then run the system
$P$ to find $k$ ones, of which we know that $n/2c$ are ``correct'' ones
present in the original string $x$. We do not have the budget to check
those $k$ ones (that would increase the degree too much). Instead, we
repeat the proccess $O(1)$ times, to get $O(1)$ different candidate
sets of $k$ positions, and take the intersection of all these sets.
This gives us a set of positions which includes the $n/2c$ real ones
plus an extra number of positions which are actually zeros
(the extra number of positions can be made an arbitrarily small
fraction on $n$).

To get the correct positions,
we just guess $n/2c$ positions out of that resulting set and
output those positions. This way, we only used $O(1)$ times the degree
of $P$ in the construction. We analyze the negative-log-probability
of success. This is $\log(1/\delta)$ for running $P$ once; it increases
by additive $\log n$ when conditioning on the number of ones increasing
to exactly $k$ (which has $\Omega(1/\sqrt{n})$ probability). It then
increases by a factor of $O(1)$ when we repeat $O(1)$ times. Finally,
it increases additively by an arbitrarily small constant fraction of $n$,
due to our guess of $n/2c$ correct positions out of slightly more positions.
The final negative-log-probability is therefore at most
$a\log(1/\delta)+a\log(n)+n/b$
where the constant $b$ can be made arbitrarily large by increasing
the constant $a$. By the lower bound on $f_{n,n/2c}$,
this must be at least $\Omega(n)$, unless the degree of $P$ is $\Omega(n)$.
Thus either $\log(1/\delta)=\Omega(n)$ or $\deg(P)=\Omega(n)$,
which gives the desired result.
\end{proof}

One corollary of this theorem is the lower bound for quantum
query complexity. This was shown in previous work using completely
different techniques (the multiplicative-weight adversary
method). In other words, we just reproved the following theorem,
though with worse constants.

\begin{theorem}[\cite{ASdW07,Spa08}]\label{thm:GroverCopies}
For $K\le n/4e$, consider the following task: given a string
of length $n$ and Hamming weight $K$, find all $K$ ones.
Then any quantum algorithm which succeeds with probability
at least $2^{-K/10}$ on this task must use
$\Omega(\sqrt{Kn})$ queries.
\end{theorem}

We now establish a lower bound for many copies of $\AndOr_n$.

\begin{theorem}\label{thm:d-direct-prod}
Let $\AndOr_n^t$ denote the search problem of solving $t$
separate copies of the search problem of $\AndOr_n$
(that is, given $t$ inputs to $\AndOr_n$, find an all-$1$ block
in each of them and output all $t$ blocks).

Assume $t\le n$ and that $n$ is sufficiently large.
There is a universal constant $C$ such that
when the input is sampled from the uniform distribution,
any polynomial system which solves this search problem with
probability at least $2^{-\frac{t}{C\log\log n}}$
must have degree at least $\Omega(t\sqrt{n}/\log\log n)$.
\end{theorem}

\begin{proof}
The proof will proceed by a reduction from \thm{DegGroverCopies}.
Suppose a system $\{p_c\}_c$ solves the search problem
$\AndOr_n^t$ with probability $q>2^{-t/C\log\log n}$ against the uniform
distribution. Recall that $\AndOr_n$ is an OR of size $n^2$
composed with an AND of size $\log_2(n/2)$, with $n$ a power of $2$,
and the goal is to output an all-$1$ block of $\log_2(n/2)$ bits;
in $\AndOr_n^t$, there are $tn^2$ blocks of size $\log_2(n/2)$
each, and our goal is to output $t$ different all-$1$ blocks
from the $t$ different batches of $n^2$ blocks.
We assume $\{p_c\}_c$ solves this task with success probability
$q$ against the uniform distribution.

We can modify the system $\{p_c\}_c$ so that it solves
a similar task: given $t$ strings of length $n^2$ each,
all of whose bits are sampled from $\mathrm{Bernoulli}(2/n)$,
find a $1$ in each of the $t$ strings. This is the same task,
except each block is replaced by a bit. To solve this task,
we can sample a fake input of length $tn^2\log(n/2)$ from
the uniform distribution conditioned on no all-$1$ block
existing; we can then plug into each variable $z_{ij}$
of the polynomial system $\{p_c\}_c$ the Boolean value
$x_i\vee y_{ij}$, where $x_i$ refers to a bit of the true
input $x\in\B^{tn^2}$ and where $y_{ij}$ refers to a bit
of the fake input $y\in\B^{tn^2\log_2(n/2)}$.
That is, instead of ``querying'' position $j$ inside block $i$,
the polynomial system now queries bit $i$ of the real input,
and if it is $0$, it queries position $j$ of the fake input.
The OR of two bits $x_i\vee y_{ij}$ can be implemented
as a linear function of $x_i$, since $y_{ij}$ not part of the input.

Observe that the simulated oracle $x_i\vee y_{ij}$ is distributed
exactly uniformly, since we first check if the AND of block $i$ evaluated
to $1$ (which happens with probability $2/n$), and if it is $0$,
we sample the block from the relevant conditional distribution.
This ensures that the system of polynomials still finds an all-$1$
block with probability $q$, so it correctly solves the new task
on inputs sampled from $\mathrm{Bernoulli}(2/n)^{\otimes tn^2}$.

By \lem{Chernoff}, the Hamming weight of each string of length
$n^2$ is in $[n,3n]$ except with probability $2^{-n/4}$.
By the union bound, all blocks have Hamming weight in $[n,3n]$
except with probability $t\cdot 2^{-n/4}<q/2$. Hence with
probability at least $q/2$, the system succeeds against
the distribution $\mathrm{Bernoulli}(2/n)^{\otimes tn^2}$
conditioned on the Hamming weight of the string being in $[tn,3tn]$.
Since this conditional distribution is a convex combination
of distributions which are each uniform over a constant Hamming weight,
there must be at least one Hamming weight $K\in[tn,3tn]$
such that the system success with probability at least $q/2$
against the uniform distribution $\mu_K$
over strings of Hamming weight $K$.

By shuffling an input string,
we get a polynomial system (of the same degree
as the original one) which, given an input of length $tn^2$
and Hamming weight $K\in[tn,3tn]$, finds $t$ ones with
probability at least $q/2$. We can repeatedly run this
(shuffling the string again each time) to find $t$ independent
ones each time; after $O(K/t)=O(n)$ iterations, we can reduce
the number of unfound ones to $K/2$, and hence after
$O(n\log\log n)$, we can reduce it to $\beta K/\log n$
for any small constant $\beta$. At that point, we can
guess the position of the remaining ones, succeeding with
probability
\[\binom{tn^2}{\beta K/\log n}^{-1}
\ge \left(\frac{\beta K}{etn^2\log n}\right)^{\beta K/\log n}
\ge \left(\frac{\beta}{en\log n}\right)^{\beta K/\log n}
\ge 2^{-O(\beta K)}.\]
The overall success probability for finding all $K$ ones
is this times $(q/2)^{O(n\log\log n)}$,
which is at least $2^{-\alpha K}$
(with $\alpha$ from \thm{DegGroverCopies}) for appropriate choice
of constants $\beta$ and $C$. Hence the degree of the resulting
system must be at least $\alpha\sqrt{tn^2K}\ge \alpha tn^{3/2}$.
This degree is also at most $O(dn\log\log n)$ where $d$
was the degree of the original system;
hence we must have $d=\Omega(t\sqrt{n}/\log\log n)$,
as desired.
\end{proof}

We immediately get the following corollary for the quantum
query complexity of this task.

\begin{corollary}\label{cor:q-direct-prod}
Let $\AndOr_n^t$ denote the search problem of solving $t$
separate copies of the search problem of $\AndOr_n$
(that is, given $t$ inputs to $\AndOr_n$, find an all-$1$ block
in each of them and output all $t$ blocks).

Assume $t\le n$ and that $n$ is sufficiently large.
There is a universal constant $C$ such that
when the input is sampled from the uniform distribution,
any quantum algorithm $Q$ which solves this search problem with
probability at least $2^{-\frac{t}{C\log\log n}}$
must make at least $\Omega(t\sqrt{n}/\log\log n)$ quantum queries.
\end{corollary}

\begin{theorem}\label{thm:Q-lb}
For $t=\log^5n$, any quantum algorithm that makes $o(\sqrt{n}\log^4 n)$ queries computes $\AndOr_{\clH,t}$ with $2^{-\Omega(\log^4n)}$ success probability against the uniform distribution.
\end{theorem}

\begin{proof}
Fix an $h$. If $x^1, \ldots, x^t$ are each distributed uniformly at random over $\B^{n^2\log(n/2)}$, then $(x^1\oplus h, \ldots, x^t\oplus h)$ is distributed uniformly at random over $\B^{tn^2\log(n/2)}$. Therefore, by \cor{q-direct-prod}, if a quantum algorithm makes $o(t\sqrt{n}/\log n) = o(\sqrt{n}\log^4n)$ queries, its success probability for computing $\AndOr_n^t$ on $(x^1\ldots h,\ldots, x^t\ldots h)$ is $2^{-\Omega(t/\log n)} = 2^{-\Omega(\log^4n)}$ when the $x^\ell$-s are uniform). These queries are to $(x^1\oplus h, \ldots, x^t\oplus h)$, but since $h$ is fixed, each such query can be implemented by a query to $(x^1,\ldots, x^t)$. Now such a quantum algorithm succeeds in computing $\AndOr_{\clH,t}$ on $(x^1,\ldots,x^t)$ if it succeeds at computing the search problem $\AndOr_n^t$ on $(x^1\oplus h, \ldots, x^t\oplus h)$ for \emph{any} $h\in \clH$. By the union bound, the success probability of a quantum algorithm computing $\AndOr_{\clH,t}$ on uniform $(x^1,\ldots, x^t)$ is thus at most $|\clH|\cdot 2^{-\Omega(\log^4 n)}$. Since $\clH$ is a set of $\log^2 n$-wise independent functions from $\B^{2\log n\log\log(n/2)}$ to $\B$, by \lem{kwise}, $|\clH| = 2^{2\log^3n\log\log(n/2)}$. Therefore, the success probability of the quantum algorithm for computing $\AndOr_{\clH,t}$ is $2^{-\Omega(\log^4n)}$.
\end{proof}

From the above theorem, \thm{adeg-ub} and \cor{jdeg-ub} we get the following two results.
\begin{corollary}\label{cor:TFNPQadeg}
There is a family of query $\mathsf{TFNP}$ problems defined on
the domain
$\B^{\poly(n)}$ such that $\adeg(f), \appdeg(f) = O(\polylog(n))$ and $\Q(f) = \tilde{\Omega}(\sqrt{n})$.
\end{corollary}

\begin{corollary}
There is a family of partial functions $f$ for which
$\adeg(f), \appdeg(f) =O(\polylog n)$ and $\Q(f)=\tOmega(n^{1/6})$,
where $n$ is a parameter such that the input size
of $f$ is at most $2^{\polylog n}$.
\end{corollary}

\begin{proof}
This follows immediately from \cor{TFNPQadeg}
and \thm{TFNPtoPartial}.
\end{proof}

\section{Adversary method for search problems}

In this section, we study adversary methods and block
sensitivity methods for general relations. Many of the relationships
that are familiar from query complexity still hold for relations,
but a few others do not work anymore. We give a new separation
showing that the approximate polynomial degree $\adeg(f)$
can be exponentially larger than all the block sensitivity
and (positive) adversary methods for a $\mathsf{TFNP}$ problem.

%\subsection{Definitions of positive adversary methods}

\subsection{Relationships between sensitivity
and adversary measures}

All the measures block sensitivity and adversary measures lower bound
randomized query complexity $\R(f)$ (the positive-weight
adversary method $\Adv(f)$ additionally lower bounds
quantum query complexity $\Q(f)$). However, the methods
vary in strength.

\paragraph{Fractional versions are stronger.}
We note the following:
\[\bs_x(f)\le \fbs_x(f)\] %,\qquad \sbs_x(f)\le \fsbs_x(f),\]
which holds since a fractional block sensitivity
with weights in $\B$ is exactly the regular block sensitivity.
From this, it follows that
\[\bs(f)\le \fbs(f),
%\qquad \sbs(f)\le \fsbs(f),
\qquad \cbs(f)\le \fcbs(f).\]
In other words, the fractional measures are all larger
than the corresponding non-fractional measures.

% \paragraph{``Strong'' block sensitivity is stronger.}
% The ``strong'' block sensitivity measures
% ($\sbs_x(f)$, $\sbs(f)$, $\fsbs_x(f)$, $\fsbs(f)$)
% are all larger than the corresponding block sensitivity
% measures ($\bs_x(f)$, $\bs(f)$, $\fbs_x(f)$, $\fbs(f)$),
% since they allow more general blocks when looking
% for a large (fractional) set of independent blocks.

\paragraph{Critical bs is incomparable to bs.}
The block sensitivity measures $\bs(f)$ and $\fbs(f)$
are formally incomparable to the critical block sensitivity
measures $\cbs(f)$ and $\fcbs(f)$. The weakness of
the critical measures is that they only look at critical
inputs, and cannot lower bound relations that have no
critical inputs (a similar definition without this restriction
will fail to lower bound $\R(f)$).
The weakness of the non-critical
block sensitivity measures is that they fail
to take into account the fact that a randomized algorithm
must ``make up its mind'' about what output symbol to give
on a specific input, and do so in a consistent way;
this intuition is made explicit in the definition of
$\cbs$ in terms of completions $f'$ of $f$.

\paragraph{$\CAdv$ is strongest.}
Of all measures defined here, $\CAdv$ is the largest.
We've already seen that $\CAdv(f)\ge \Adv(f)$,
and it is easy to see that a feasible solution to the
minimization problem of $\CAdv(f)$ gives a feasible solution
to the minimization problem of $\fbs(f)$ (the fractional certificates
version). This means that $\CAdv(f)\ge\fbs(f)\ge\bs(f)$.
It also holds that $\CAdv(f)\ge\fcbs(f)/2$; this was shown
in \cite{ABK21} for partial functions; we prove it here for
all relations.

\begin{theorem}
$\CAdv(f)\ge\fcbs(f)/2$.
\end{theorem}

\begin{proof}
Let $\{m_x(i)\}$ be an optimal solution to the minimization
problem of $\CAdv(f)$. We define a completion $f'$ of $f$
as follows. For inputs $x$ and $y$, define the
asymmetric distance $d(x,y)=\sum_{i:x_i\ne y_i} m_y(i)$.
For critical $y$, define $f'(y)$ to be the unique symbol in $f(y)$.
For each non-critical $x$, find the closest critical $y$
(minimizing $d(x,y)$) subject to the condition that
$f'(y)\in f(x)$.
Define $f'(x)=f'(y)$ to get the completion $f'$.

Pick any critical $y$. It remains to upper bound
$\fbs_y(f')$. We pick the solution $c(i)=2m_y(i)$
for the minimization (fractional certificate) version
of $\fbs_y(f')$. The objective value of this is at most
$2\CAdv(f)$. We must check feasibility. Let $B$ be any
sensitive block for $y$ with respect to $f'$, and let
$x=y^B$. Then $f'(x)\ne f'(y)$.
If $f'(y)$ is not an output symbol in $f(x)$,
then $f(y)\cap f(x)=\emptyset$, and the constraint
of the $\CAdv(f)$ problem applies to give us
\[\sum_{i\in B} m_y(i)\ge
\sum_{i:x_i\ne y_i}\min\{m_y(i),m_x(i)\}\ge1,\]
so feasibility holds. Now suppose $f(y)$ is an output symbol
in $f(x)$. Let $z$ be the critical input that minimized
$d(x,z)$, so that $f'(x)=f'(z)$. Then
$f'(z)=f'(x)\ne f'(y)$, and since $y$ and $z$ are critical,
we have $f(y)\cap f(z)=\emptyset$. This means
the constraint of $\CAdv(f)$ applies to $(y,z)$.
We now write
\begin{align*}
\sum_{i\in B} c(i) &= 2\sum_{i:x_i\ne y_i} m_y(i)
\\&= 2d(x,y)
\\&\ge d(x,y)+d(x,z)
\\&= \sum_{i:x_i\ne y_i} m_y(i)+\sum_{i:x_i\ne z_i} m_z(i)
\\&\ge \sum_{i:y_i\ne z_i} \min\{m_y(i),m_z(i)\}
\\&\ge 1.
\end{align*}
Thus feasibility holds, completing the proof.
\end{proof}

\paragraph{Classical and quantum adversaries are quadratically
related.}
It was shown in \cite{ABK21} that $\CAdv(f)\le 2\Adv(f)^2$
for all partial functions; the proof extends with no
modifications to all relations $f$.

\paragraph{(Critical) bs lower bounds approximate degree.}
We show this in the following theorem.

\begin{theorem}
For all $\eps<1/2$ and all relations $f$,
\[\adeg_\eps(f)=\Omega(\sqrt{(1-2\eps)\fcbs(f)}),\]
\[\adeg_\eps(f)=\Omega(\sqrt{(1-2\eps)\fbs(f)}).\]
\end{theorem}

\begin{proof}
We start with $\fcbs(f)$.
The proof is a simple adaptation of a similar result in
\cite{ABK21} for partial functions.
Let $\{p_c\}$ be a system of polynomials minimizing 
$\adeg_\eps(f)$. For each input $x$,
define $f'(x)$ to be the symbol $c\in f(x)$ maximizing
$p_c(x)$. Note that if $x$ is critical and $c'$
is its unique correct output symbol, then by the constraints
on the polynomial system, $p_{c'}(x)\ge 1-\eps >1/2$,
and since $\sum_{c} p_c(x)\le 1$, it must be the case
that $f'(x)=c'$.

Now pick critical $y$ arbitrarily. We wish to upper bound
$\fbs_y(f')$. Let $c$ be the unique output symbol of $y$,
and consider just the polynomial $p_c(x)$.
We know that $p_c(y)\ge 1-\eps$, and we know that
if $f'(x)\ne f'(y)$, then $p_c(x)\le 1/2$ (since the symbol $c$
does not maximize $p_c(x)$, and the sum of the system at $x$
is at most $1$). Then the polynomial $p(x)=(2/(3-2\eps))p_c(x)$
is bounded in $[0,1]$ on all inputs, and
has the following properties: on the input $y$,
it outputs at least $(1+\gamma)/2$, where
$\gamma=(1-2\eps)/(3-2\eps)$; on the inputs $x$ such that
$f'(x)\ne f'(y)$, it outputs at most $(1-\gamma)/2$.
Therefore, if we let $g$ be the function with $g(y)=1$
and $g(x)=0$ when $f'(x)\ne f'(y)$ (with all other inputs
not in the promise), then $p$ computes $g$ to bias $\gamma$
(error $1/(3-2\eps)$). It is easy to see that
$\fbs_y(f')=\fbs_y(g)$. A result of \cite{ABK21}
gives $\adeg_{\eps'}(g)\ge
\frac{\sqrt{2}}{\pi}
\sqrt{\frac{1-2\eps'}{1-\eps'}}\sqrt{\fbs(g)}$,
from which we get 
\[\adeg_\eps(f)\ge
\frac{1}{\pi}\sqrt{\frac{1-2\eps}{1-\eps}}\sqrt{\fcbs(f)},\]
as desired.

For $\fbs(f)$, the idea is similar: fix the system
of polynomials $\{p_c\}$ as before, and fix any input $x$.
Consider the polynomial $p(x)=\sum_{c\in f(x)} p_c(x)$.
Then $p(x)\ge 1-\eps$, and if $f(y)\cap f(x)=\emptyset$,
then we must have $p(y)\le \eps$. Let $g$
be the partial function with $g(x)=1$ and $g(y)=0$
if $f(y)\cap f(x)=\emptyset$; then $p$ computes $g$ to error
$\eps$, and $\fbs_x(f)=\fbs_x(g)$. The same $\fbs$
lower bound from \cite{ABK21} now gives
\[\adeg(f)\ge
\frac{\sqrt{2}}{\pi}\sqrt{\frac{1-2\eps'}
    {1-\eps'}}\sqrt{\fbs(f)},\]
as desired.
\end{proof}

\subsection{Separating polynomials from adversaries}

Define the following $\mathsf{TFNP}$ search problem.

\begin{definition}[Pigeonhole problem]
For $n$ a power of $2$, define pigeonhole search problem
$\tPH_n$ on domain $\B^{n\log(n/2)}$ by dividing
the input into $n$ blocks of $\log(n/2)$ bits each,
and having the valid certificates be any
two identical blocks. In other words,
given $n$ blocks with $n/2$ possible symbols in each block
(with the symbols represented in binary using $\log(n/2)$ bits),
the task is to find two blocks with the same symbol;
such a pair is guaranteed to exist by the pigeonhole principle.
\end{definition}

It is easy to see that $\tPH_n$ is a valid query search
problem in $\mathsf{TFNP}$; each input has a valid certificate,
and all the certificates are of size $2\log(n/2)$, which
is logarithmic in the input size $n\log(n/2)$.

\begin{lemma}
$\CAdv(\tPH_n)=O(\log n)$.
\end{lemma}

\begin{proof}
We just need to find a feasible solution to the minimization
problem in the definition of $\CAdv(\tPH_n)$.
For each $x\in\Dom(\tPH_n)$ and each $i\in[n\log(n/2)]$,
set $m_x(i)=\frac{2}{n}$. Then
the objective value is $\sum_i m_x(i)=2\log(n/2)$.
To check feasibility, note that if $(x,y)\in\Delta(\tPH_n)$,
then $x$ and $y$ must have Hamming distance at least $n/2$.
This is because if the Hamming distance were less than $n/2$,
there would be at least $n/2+1$ different blocks which
are the same in $x$ and in $y$, and since there are only $n/2$
possible symbols in a block, the pigeonhole principle
would imply that $x$ and $y$ share a pair of identical
blocks, and hence share a certificate; that would contradict
$(x,y)\in\Delta(\tPH_n)$.

This means that for all $(x,y)\in\Delta(f)$, we have
$\sum_{i:x_i\ne y_i}\min\{m_x(i),m_y(i)\}\ge (n/2)(2/n)=1$,
and feasibility holds. Hence $\CAdv(\tPH_n)\le 2\log(n/2)$,
as desired.
\end{proof}

\begin{theorem}
$\adeg(\tPH_n)=\tOmega(n^{1/3})$.
\end{theorem}

\begin{proof}
Let $n$ be a power of $2$.
We proceed via reduction from the collision lower bound.
In the collision problem, we are given
a string in $[n]^n$ which is promised to be either a permutation
(where each symbol occurs exactly once) or a $2$-to-$1$ function
(where half the symbols in $[n]$ occur exactly twice each).
Distinguishing these two cases to bounded error
cannot be done with a bounded polynomial
of degree less than $\Omega(n^{1/3})$ \cite{AS04,Amb05,Kut05}.
We will show how a system of polynomials $\{p_c\}$
computing $\tPH_n$ can be used to construct a bounded
polynomial for collision whose degree is about the same
as the degree of the system; this will prove that
$\adeg(\tPH_n)=\Omega(n^{1/3})$.

Before we start the reduction, we impose a certain symmetry
on a system of polynomials $\{p_c\}$ for $\tPH_n$.
For each polynomial in such a system, we can define
$p_c^\pi(x)=p_c(x_\pi)$, where $\pi\in S_n$ is a permutation
of the $n$ blocks and $x_\pi$ is the input $x$ with the blocks
shuffled according to $\pi$. It is not hard to see that
if $\{p_c\}$ solves $\tPH_n$, then so does $\{p_{c_\pi}^\pi\}$,
where $c_\pi$ is the certificate $c$ with blocks shuffled
according to $\pi$.
Indeed,
the feasibility condition for $\{p_{c_\pi}^\pi\}$ at $x$,
$\sum_{c\subseteq x} p_{c_\pi}^\pi(x)\in [2/3,1]$
is the same as
$\sum_{c_\pi\subseteq x_\pi} p_{c_\pi}(x_\pi)\in [2/3,1]$,
which is the feasibility condition for $\{p_c\}$ at $x_\pi$.
Therefore, by letting $\overline{p}_c=\frac{1}{n!}\sum_\pi p_{c_{\pi}}^\pi$, we get
\[\sum_{c\subseteq x} \overline{p}_c(x)
=\frac{1}{n!}\sum_\pi\sum_{c\subseteq x} p^\pi_{c_{\pi}}(x)
\in[2/3,1].\]
This new solution $\{\overline{p}_c\}$ has the property
that $\overline{p}_{c_\pi}^\pi=\overline{p}_c$ for all $\pi$.
In particular, for any input $x$, if $\pi$ permutes
only identical blocks in $x$, then $x^\pi=x$,
and hence $\overline{p}_c(x)=\overline{p}_{c_\pi}^\pi(x)
=\overline{p}_{c_\pi}(x)$. Therefore, if $c$ and $c'$
are two certificates for $x$ such that the contents of each
block in $c$ and in $c'$ are all the same, then
$\overline{p}_c(x)=\overline{p}_{c'}(x)$.
From now on, we will assume without
loss of generality that $\{p_c\}$ is a symmetrized solution
of this form (replacing $\{p_c\}$ by $\{\overline{p}_c\}$ if
necessary).

We now proceed with the reduction.
Let $y\in[n]^n$ be an input to collision. Pick a
$2$-to-$1$ function $g\colon[n]\to[n/2]$ and implement
it via $\log(n/2)$ polynomials of degree at most $\log n$ each:
the $i$-th bit of $g(\ell)$ is a polynomial $q_i$ of degree $\log n$
in the bits of $\ell\in\B^{\log n}$. Now, take a system
of polynomials $\{p_c\}$ computing $\tPH_n$ to bounded error,
and compose each of these polynomials with $n$ copies of $g$:
that is, for each polynomial $p_c$, for each of the $n$ blocks
of the input, find the $\log(n/2)$ variables in $p_c(x)$ that
come from the $n$-th block, and into those variables plug
in the polynomials $q_i$ applied to the $i$-th symbol in
the input $y$ to collision. The effect of this is to use $g$
to map the $n$ symbols of $y$ into $n$ symbols in $[n/2]$
via the $2$-to-$1$ function $g$, and then apply $p_c$ to the result.
This increases the degree of each polynomial $p_c$ in the system
by a factor of $\log n$. Call the new polynomials $p_c^g$.

The polynomials $\{p_c\}$ now ``return a certificate''
for this new input $x=g^n\circ y$. Concretely, this means
that for each $y\in[n]^n$, we have
$\sum_{c\subseteq x} p_c^g(y)\ge 2/3$ where $x=g^n\circ y$.
The condition $c\subseteq x$ says that $c$ consists of two
identical blocks of $\log(n/2)$ bits in $x$.
Note also that the symmetry property we imposed on $\{p_c\}$
ensures that $p_c(x)$ depends only on $x$ and on the value
of the two blocks revealed by $c$, but not on the locations
of the blocks. A certificate $c$ corresponds
to two symbols in $y$, say $y_i,y_j\in[n]$, such that
$g(y_i)=g(y_j)$, and the value of $p_c^g(y)$ depends on $y$
and on $g(y_i)$ but not on $i$ and $j$.

There are two ways for $g(y_i)=g(y_j)$ to happen:
either $y_i=y_j$, or else $y_i=\ell\ne\ell'=y_j$, but
$g(\ell)=g(\ell')$ (recall $g$ is a $2$-to-$1$ function).
If $y$ is a permutation, only the latter can happen.
If $y$ is a $2$-to-$1$ function, then for a given value
of $g(y_i)$, there are either zero blocks with
that $g$-value, or there are $2$ blocks with that
$g$-value corresponding to $i,j$ with $y_i=y_j$,
or else there are $4$ blocks with that $g$-value
corresponding to $i,j,a,b$ with $y_i=y_j\ne y_a=y_b$
and $g(y_j)=g(y_a)$. Those $4$ blocks correspond
to $\binom{4}{2}=6$ certificates, of which
two correspond to a true collision ($y_i=y_j$
and $y_a=y_b$). Since $p_c^g(y)$ depends only
on the $g$-value used by the certificate $c$,
we conclude that when $y$ is a $2$-to-$1$ function,
at least $1/3$ of the weight
in $\sum_{c\subseteq x} p_c^g(y)$ corresponds
to certificates $c$ representing true collisions.

Let $r_c(y)$ be the polynomial of degree $O(\log n)$
that checks for a collision in $y$ at the two blocks
revealed by the certificate $c$; $r_c(y)$ returns
$1$ if the collision occurs, and $0$ otherwise.
Consider the polynomial $s(y)=\sum_c p_c^g(y)r_c(y)$,
which has degree $O(\log^2 n)$ larger than the maximum
degree of $\{p_c\}$. If $y$ is $1$-to-$1$, it has no
true collisions, so $r_c(y)=0$ for all $c$,
and $s(y)=0$. For arbitrary $y$, we know that
$r_c(y)=0$ unless $c\subseteq x$, where $x=g^n\circ y$.
Thus $s(y)=\sum_{c\subseteq x} p_c^g(y)r_c(y)$,
which is in $[0,1]$ for all $y$, meaning $s(y)$ is bounded.
Finally, if $y$ is $2$-to-$1$, we know that
$\sum_{c\subseteq x} p_c^g(y)\ge 2/3$ and that at least
$1/3$ of this weight corresponds to true collisions
(for which $r_c(y)=1$), so we have $s(y)\ge 2/9$.

The polynomial $1-(1-s(y))^5$ is still bounded in $[0,1]$,
still returns $0$ when $y$ is $1$-to-$1$, but now returns
at least $0.7$ when $y$ is $2$-to-$1$. Therefore, it solves
the collision problem, and has degree $\Omega(n^{1/3})$.
Since its degree was at most $O(\log^2 n)$ times
$\adeg(\tPH_n)$, we conclude the latter is
$\tOmega(n^{1/3})$, as desired.
\end{proof}

%\subsection{Definitions of block sensitivity measures}

\subsection{Discussion}\label{sec:discussion}

We have shown the following: 
\[\bs(f)\le\fbs(f)\le \{\CAdv(f),\adeg(f)^2\},\]
\[\cbs(f)\le\fcbs(f)\le\{\CAdv(f),\adeg(f)^2\}.\]
We also saw that $\CAdv(f)\approx\Adv(f)$ (up to a quadratic
factor), and $\adeg(f)$ can be exponentially larger
than $\CAdv(f)$ for a $\mathsf{TFNP}$ problem.
In \cite{ABK21}, it was additionally shown that for partial functions,
$\adeg(f)=\Omega(\sqrt{\CAdv(f)})$, by showing
that $\CAdv(f)\approx \fcbs(f)$ for all partial functions.
However, we do not know how to generalize this to relations.
This raises the following interesting open problem.

\begin{open}
Can we have $\Adv(f)\gg \adeg(f)$ for a relation $f$?
That is, is there a relational problem for which the adversary
method is exponentially larger than the polynomial method?
\end{open}

We know from \cite{ABK21} that such a separation cannot happen
for partial functions, but this problem is open for
relations and for $\mathsf{TFNP}$ search problems.
% By \thm{IC-rprt}, the distributional relaxed partition bound is a lower bound on the $\IC$ for the corresponding distribution. We therefore show that this lower bound can be exponentially weak.

% \begin{lemma}
% For any relation $f$, we have $\sbs(f)\ge \cbs(f)$
% and $\fsbs(f)\ge \fcbs(f)$.
% \end{lemma}

% \begin{proof}
% Start with a feasible set of solutions $\{c_x(i)\}$ 
% to the minimization version of $\fsbs(f)$. For each input
% $x\in\Dom(f)$, consider all the critical inputs
% $y$, and find the one which minimizes
% $\sum_{i:x_i\ne y_i} c_x(i)$. Set $f'(x)=f(y)$ for
% that minimizing critical input $y$.
% Now for any critical $z$ and any sensitive block
% $B$ of $z$ with respect to $f'$, we have
% $f'(z^B)\ne f'(z)=f(z)$. Let $x=z^B$, and let $y$
% be the critical input to $f$ which is closest to $x$,
% so that $f'(x)$ was chosen to be $f(y)$. Then
% $z$ and $y$ are both critical, with $f(z)\ne f(y)$.
% Let $B'=B\cup\{i:x_i\ne y_i\}$. Then $x$ agrees with
% $z$ outside $B$ and agrees with $y$ outside
% $\{i:x_i\ne y_i\}$, and so agrees with both outside $B'$.
% This means $B'$ is semisensitive for $x$, and hence
% $\sum_{i\in B'} c_x(i)\ge 1$. 

% \end{proof}

%\section{Application: separating information cost and relaxed partition bound}
\section{Separating information cost and relaxed partition bound}

\subsection{Communication complexity measures}
In the communication model, two parties, Alice and Bob, are given inputs $x\in\X$ and $y\in\Y$ respectively, and the task is to jointly compute a relation $f\subseteq \X\times\Y\times\Sigma$ by communicating with each other. In other words, on input $(x,y)$,
Alice and Bob must output a symbol
$s\in\Sigma$ such that $s\in f(x,y)$.
Without loss of generality, we can assume Alice sends the first message, and Bob produces the output of the protocol.

In the classical randomized model, Alice and Bob are allowed to use shared randomness $R$ (and also possibly private randomness $R_A$ and $R_B$) in order to achieve this. The cost of a communication protocol $\Pi$, denoted by $\CC(\Pi)$ is the number of bits exchanged between Alice and Bob. The randomized communication complexity of a relation $f$ with error $\eps$, denoted by $\R^\CC_\eps(f)$, is defined as the minimum $\CC(\Pi)$ of a randomized protocol $\Pi$ that computes $f$ with error at most $\eps$ on every input; we omit writing $\eps$ when $\eps=1/3$. $\D^\CC_\eps(f,\mu)$ is the minimum $\CC(\Pi)$ of a deterministic protocol that that computes $f$ with distributional error at most $\eps$ over $\mu$.

\paragraph{Information complexity.} The information complexity of a protocol with inputs $X,Y$ according to distribution $\mu$ on $\X\times\Y$, shared randomness $R$ and transcript $\Pi$ is given by
\[ \IC(\Pi,\mu) = I(X:\Pi|YR)_\mu + I(Y:\Pi|XR)_\mu.\]
For any $\mu$ we have, $\IC(\Pi,\mu) \leq \CC(\Pi)$. The distributional information cost with $\eps$ error $\IC_\eps(f, \mu)$ of a relation $f$ is defined as the minimum information cost of a protocol computing $f$ up to $\eps$ error with respect to $\mu$, and the information cost of $f$ with $\eps$ error is the maximum value of $\IC_\eps(f,\mu)$ over all $\mu$. We naturally have, $\IC_\eps(f,\mu) \leq \D^\CC_\eps(f,\mu)$, and $\IC_\eps(f) \leq \R^\CC_\eps(f)$. Like with other measures, we drop the $\eps$ from the notation for $\IC$ when $\eps=1/3$.

The following theorem is obtained by inspecting the proof of the upper bound in Theorem 6.3 in \cite{BR14}; it is stated for functions in \cite{BR14}, but the same proof holds for relations as well.
\begin{theorem}[\cite{BR14}]\label{thm:amortized-IC}
For any relation $f \subseteq \B^n\times\Sigma$, let $f^k$ denote $k$ copies of $f$. Then for $k=\Omega(n^2)$, any $\eps \in (0,1)$ and any distribution $\mu$ on the inputs of $f$,
\[ \frac{\D^\CC_{2\eps}(f^k,\mu^{\otimes k})}{k} \leq \IC_\eps(f, \mu) + O(\log(1/\eps)).\]
\end{theorem}

\paragraph{Relaxed partition bound.}
For a relation $f$ on $\X\times\Y$, a combinatorial rectangle $R$ is a subset of $\X\times\Y$ of the form $\X'\times\Y'$, where $\X'\subseteq \X$ and $\Y'\subseteq Y$. The relaxed partition bound with error $\eps$ $\rprt_\eps(f)$ for the relation $f\subseteq\X\times\Y\times\Sigma$ is defined as the optimal value of the following linear program, in which the optimization variables are indexed by $s\in \Sigma$, and rectangles $R$:
\begin{equation}\label{eq:rprt}
\begin{array}{lrll}
\min & \sum_{s\in \Sigma,R}w_{s,R} & & \\
\text{s.t.} & \sum_{s\in \Sigma, R \ni (x,y)}w_{s,R} & \leq 1 & \forall (x,y) \in \X\times\Y \\
& \sum_{s\in f(x,y), R \ni (x,y)} w_{s,R} & \geq 1 -\eps & \forall (x,y) \in \X\times\Y \\
& w_{s,R} & \geq 0 & \forall s\in \Sigma, R 
\end{array}.
\end{equation}
As usual, we omit the subscript $\eps$ when $\eps=1/3$. There is a distributional version of the relaxed partition bound, which we will not be defining, and it holds that $\rprt_\eps(f) = \max_\mu \rprt_\eps(f,\mu)$, where $\rprt_\eps(f,\mu)$ is the relaxed partition bound with error $\eps$ with respect to distribution $\mu$ on the inputs.

It was shown in \cite{KLL+15} that $\rprt(f,\mu)$ is a lower bound on $\IC(f,\mu)$.\footnote{The claim made in \cite{KLL+15} was only for (not necessarily boolean) functions, but an inspection of the proof shows it also holds for relations.} Moreover, $\rprt(f)$ is an upper bound on the two-sided non-negative approximate rank of $f$, also known as the two-sided smooth rectangle bound for $f$. So any separation between $\IC$ and $\rprt$ is also a separation between $\IC$ and two-sided non-negative approximate rank.

% The following relation between information complexity and the relaxed partition bound was proved in \cite{KLL+15} for functions, but an inspection shows it also holds for relations.
% \begin{theorem}[\cite{KLL+15}]\label{thm:IC-rprt}
% For a relation $f\subseteq \X\times\Y\times\Sigma$, any distribution $\mu$ on $\X\times\Y$, and $\eps \in (0,1)$,
% \[ \IC_\eps(f,\mu) = \Omega(\eps^2\log(\rprt_{4\eps}(f,\mu)) - \log|\Sigma|).\]
% \end{theorem}

\subsection{Lifting randomized query complexity to information cost}
For $x,y \in \B^b$, $\IP_b(x,y)=x_1y_1+\ldots+x_by_b \mod 2$. For $x\in \B^{bn}$, $i\in [n]$ and $j\in[b]$, let $x^i_{j}$ denote the $j$-th bit in the $i$-th block of $x$, where each block is $b$ bits. We'll also use $x^i$ to refer to the $i$-th block as a whole. For a relation $f \subseteq \B^n\times\B^n\times\Sigma$, the composed relation $f\circ \IP^n_b$ on $\B^{bn}\times\B^{bn}\times\Sigma$ is given by
\[ f\circ \IP^n_b(x,y) = f(\IP(x^1_{1}\ldots x^1_{b}, y^1_{1}\ldots y^1_{b}), \ldots, \IP(x^n_{1}\ldots x^n_{b}, y^n_{1}\ldots y^n_{b})). \]
Note that with this definition, $f^k\circ \IP^{nk}_b = (f\circ \IP^n_b)^k$, i.e., $k$ copies of the function $f\circ \IP^n_b$.

Let $G$ denote $\IP_b^n$, for $n,b$ implied from context. For a distribution $\mu$ on $\B^n$, let $\mu^G$ denote the distribution on $\B^{bn}\times\B^{bn}$ denote the distribution generated by the following procedure:
\begin{enumerate}
\item Sample $z\sim \mu$
\item Sample $(x,y) \sim G^{-1}(z)$.
\end{enumerate}
It is not difficult to see that if $\mu$ is the uniform distribution on $\B^n$, then $\mu^G$ is the uniform distribution on $\B^{bn}\times\B^{bn}$.

\begin{theorem}[\cite{GPW20, CFK+19}]\label{thm:lifting}
For $b=\Omega(\log n)$, any $\eps \in (0,1)$, and any distribution $\mu$ on $\B^n$, $\D^\CC_\eps(f\circ\IP^n_b, \mu^G) = \Omega(\D_{2\eps}(f,\mu)\cdot b)$.
\end{theorem}

\begin{theorem}[\cite{JKS10}]\label{thm:query-ds}
For any relation $f\subseteq \B^n\times\Sigma$, any $\eps \in (0,1)$, and any distribution $\mu$ on $\B^n$,
\[ \D_\eps(f^k) = \Omega(\eps^2k\cdot\D_{3\eps}(f,\mu)).\]
\end{theorem}

\begin{theorem}\label{thm:IC-lifting}
For any relation $f \subseteq \B^n\times\Sigma$, $b = \Omega(\log n)$, and any $\eps \in (0,1)$,
\[ \IC_\eps(f\circ \IP^n_b, \mu^G) = \Omega(\D_{6\eps}(f,\mu)\cdot b) - O(\log(1/\eps)).\]
\end{theorem}
\begin{proof}
We apply \thm{lifting} to $f^k$ for $k=\Theta(n^2)$, to get for $b=\Omega(\log(kn)) = \Omega(\log n)$,
\[ \D_\eps(f^k\circ \IP^{kn}_{b}, \mu^G) = \Omega(\D_{2\eps}(f^k,\mu)\cdot\log (kn)) = \Omega(\D_{2\eps}(f^k,\mu)\cdot\log n) \]
By \thm{query-ds}, we then have,
\[ \D_\eps(f^k\circ \IP^{kn}_{b}, \mu^G) = \Omega(\eps^2k\cdot \D_{6\eps}(f, \mu)\cdot\log n).\]
Since $f^k\circ \IP^{kn}_{b} = (f\circ \IP^n_{b})^k$, and we have taken $k=\Theta(n^2)$, we have by \thm{amortized-IC},
\[ \IC(f\circ\IP^n_{b}, \mu^G) + O(\log(1/\eps)) = \Omega(\eps^2\D_{6\eps}(f,\mu)\cdot \log n),  \]
which completes the proof.
\end{proof}

Let $\AndOr_{\clH,t}$ be the function from \thm{Q-lb}, which is defined on inputs in $\B^{tn^2\log(n/2)}$. The above theorem then has the following implication for the lifted versions of $\AndOr_{\clH,t}$, and $(\AndOr_{\clH,t})_{\fcn}$.
\begin{theorem}\label{thm:IC-lb}
For $t=\log^5n$, and  for $b = \Theta(\log(tn^2\log(n/2))$,
\begin{align*}
\IC(\AndOr_{\clH,t}\circ \IP^{tn^2\log(n/2)}_b) & = \Omega(\sqrt{n}\log^5n), \\
\IC((\AndOr_{\clH,t})_{\fcn}\circ\IP^{tn^2\log(n/2}_b) & = \Omega(\sqrt{n}\log^5n).
\end{align*}
\end{theorem}
\begin{proof}
\thm{Q-lb} implies that any distributional complexity of $\AndOr_{\clH,t}$ against the uniform distribution is $\Omega(\sqrt{n}\log^4n)$. \thm{IC-lifting} thus gives the lower bound on $\IC$ of $\AndOr_{\clH,t}\circ \IP^{tn^2\log(n/2)}$ against the uniform distribution, which is a lower bound on the worst-case $\IC$. Similarly, the second statement follows from \thm{TFNPtoPartial} for the partial function version of $\AndOr_{\clH,t}$.
\end{proof}

\subsection{Upper bound on relaxed partition bound}
We first observe that for $f\subseteq \B^n\times\Sigma$, its conical junta degree can be written as an optimization problem. Let $C$ denote a conjunction, $|C|$ denote the width of the conjunction (the number of variables in it), and we say $x\in C$ if $x$ satisfies the conjunction. The set of all $x$ satisfying a conjunction is a subcube. Then we have,
\begin{equation}\label{eq:jdeg}
\begin{array}{rlrll}
\adeg\!{}_\eps^{+2}(f) = & \min & d & & \\
 & \text{s.t.} & \sum_{s\in\Sigma, C: |C| \leq d, C\ni x}w_{s,C} & \leq 1 & \forall z \in \B^n \\
 & & \sum_{s\in f(x), C: |C| \leq d, C \ni x} w_{s,C} & \geq 1-\eps & \forall z \in \Dom(f) \\
 & & w_{s,C} & \geq 0 & \forall s \in \Sigma, |C| \leq d.
\end{array}
\end{equation}
Using this formulation, we can prove an upper bound the relaxed partition bound for $\AndOr_{\clH,t}\circ \IP^{tn^2\log(n/2)}_b$.
\begin{theorem}\label{thm:rprt-ub}
For $t=\polylog(n)$, $b= \Theta(\log(tn^2\log(n/2))$,
\begin{align*}
\rprt(\AndOr_{\clH,t}\circ\IP^{tn^2\log(n/2)}_b) & = 2^{O(t\log^3n)} \\
\rprt((\AndOr_{\clH,t})_{\fcn}\circ\IP^{tn^2\log(n/2)}_b) & = 2^{O(t\log^3n)}
\end{align*}
\end{theorem}
\begin{proof}
We'll only give the proof for $\rprt(\AndOr_{\clH,t}\circ\IP^{tn^2\log(n/2)}_b)$. The proof for the partial function version follows in essentially the same way by using \thm{TFNPtoPartial} to upper bound the conical junta degree of $(\AndOr_{\clH,t})_{\fcn}$.

Let $m=tn^2\log(n/2)$. \cor{jdeg-ub} implies that there exists a solution to \eq{jdeg} for $\AndOr_{\clH,t}$ with $d=O(t\log^2n)$.
Now there are only $\binom{m}{d}\cdot 2^d = n^{O(d)}$ many conjunctions of width $d$ on the inputs to $\AndOr_{\clH,t}$ (we choose $d$ variables out of $m$ to put in the conjunction, and for each variable, we choose whether to include the variable or its negation). Moreover, there are $2^{2t\log n}$ many possible outputs for $\AndOr_{\clH,t}$, since each of the $t$ blocks of the input, has at most $n^2$ many certificates. The condition $\sum_{s\in\Sigma, C: |C| \leq d, C\ni x}w_{s,C}$ implies that each of the weights $w_{s,C}$ for conjunctions of width $d$ is at most 1. We can assume the weights on conjunctions which have width more than $d$ are 0. Therefore,
\begin{equation}\label{eq:weights-ub}
\sum_{s\in \Sigma, C} w_{s,C} \leq n^{O(d)}\cdot 2^{2t\log n} = 2^{O(t\log^3 n)},
\end{equation}
using the upper bound on $d$.

Now consider a conjunction $C$ (on which $w_{s,C}$ puts non-zero weight) on $z^1\ldots z^m$, with $z^1\ldots z^m = \IP(x^1,y^1)\ldots \IP(x^m,y^m)$. Each conjunction involves $d$ many $z^i$-s, and requires each $z^i$ to take a particular value. For each such conjunction $C$, we define a collection $\clR(C)$ of $2^{b}$ rectangles on $\B^{bm}\times\B^{bm}$ as follows: for every $z^i$ in the conjunction, there will be a rectangle corresponding to every possible value of $(x^i,y^i)$ such that $\IP(x^i,y^i)$ evaluates to the value of $z^i$ required by the conjunction. There are $2^{2b}$ many possible values of $(x^i,y^i)$, and half of them evaluate to 0, and the other half to 1. Since there are $d$ of these, there are $2^{db}$ many rectangles overall. We'll construct a weight scheme for $\rprt(\AndOr_{\clH,t}\circ \IP^m_b)$ as follows:
\[ v_{s,R} =  w_{s,C} \quad \text{iff } R \in \clR(C).\]
We need to show that this weight scheme satisfies \eq{rprt} for $\AndOr_{\clH,t}\circ\IP^m_b$.

First we notice that any input $(x,y)$ to $\AndOr_{\clH,t}\circ\IP^m_b$ belongs to exactly one rectangle in $\clR(C)$ if $\IP^m_b(x,y)$ satisfies $C$, and belongs to no rectangles in $\clR(C)$ otherwise. Therefore,
\begin{align*}
\sum_{s\in\Sigma,R\ni (x,y)} v_{s,R} & = \sum_{s\in \Sigma}\sum_{C, R\in \clR(C): R \ni (x,y)}w_{s,R} \\
 & = \sum_{s\in \Sigma}\sum_{C: \IP^m_b(x,y) \in C}w_{s,C} \\
 & = \sum_{s\in \Sigma, C\ni z} w_{s,C} \leq 1,
\end{align*}
where $z = \IP^m_b(x,y)$, and the last inequality follows from \eq{jdeg} for $z$. By a similar argument, we get $\sum_{s\in \AndOr_{\clH,t}\circ \IP^m_b(x,y), R \ni (x,y)}v_{s,R} \geq \frac{2}{3}$. Finally,
\begin{align*}
\sum_{s\in\Sigma,R}v_{s,R} & = \sum_{s\in \Sigma}\sum_{C, R \in \clR(C)}w_{s,C} \\
 & = \sum_{s\in\Sigma, C}|\clR(C)|\cdot w_{s,C} \leq 2^{O(t\log^3n)}\cdot 2^{db},
\end{align*}
using \eq{weights-ub} and the upper bound on $|\clR(C)|$. Putting in the values of $d$ and $b$, we get the required upper bound on the relaxed partition bound.
\end{proof}

The above theorem and \thm{IC-lb} immediately give the following corollary.
\begin{corollary}
There exists a family of $\mathsf{TFNP}$ problems $f$, and partial functions $f_{\fcn}$ for communication, with inputs in $\B^{\poly(n)}\times\B^{\poly(n)}$ such that
\begin{align*}
\IC(f) = \Omega(\sqrt{n}\log^5n), & \quad \rprt(f) = O(\log^8n), \\
\IC(f_{\fcn}) = \Omega(\sqrt{n}\log^5n), & \quad \rprt(f_{\fcn}) = O(\log^8n).
\end{align*}
\end{corollary}
The above corollary implies that there is some distribution for which the distributional versions of $\IC$ and $\rprt$ are exponentially separated for the relevant $\mathsf{TFNP}$ problem and partial function. The result for TFNP proved here will be subsumed by \thm{QIC-rprt} which we prove in the next section.

\section{Separating quantum information cost and relaxed partition bound}

\subsection{Quantum communication complexity measures}
%In the communication model, two parties, Alice and Bob, are given inputs $x\in\X$ and $y\in\Y$ respectively, and the task is to jointly compute a relation $f\subseteq \X\times\Y\times\Sigma$ by communicating with each other. In other words, on input $(x,y)$,
%Alice and Bob must output a symbol
%$s\in\Sigma$ such that $s\in f(x,y)$.
%Without loss of generality, we can assume Alice sends the first message, and Bob produces the output of the protocol.

In the quantum communication model with pre-shared entanglement, Alice and Bob are allowed to share some entanglement beforehand, and communicate with quantum messages in order to achieve this. The cost of a communication protocol $\Pi$, denoted by $\QCC(\Pi)$ is the number of qubits exchanged between Alice and Bob. The quantum communication complexity of a relation $f$ with error $\eps$, denoted by $\Q^{\CC,*}_\eps(f)$, is defined as the minimum $\QCC(\Pi)$ of a quantum protocol $\Pi$ with pre-shared entanglement that computes $f$ with error at most $\eps$ on every input; we omit writing $\eps$ when $\eps=1/3$. $\Q^{\CC,*}_\eps(f,\mu)$ is the minimum $\QCC(\Pi)$ of a quantum protocol that that computes $f$ with distributional error at most $\eps$ over $\mu$.

\paragraph{Quantum information cost.} To define quantum information cost, first we will describe a quantum communication protocol $\Pi$ with explicit names for registers. Alice and Bob initially share an entangled state on registers $A_0B_0$, and they get inputs $x$ and $y$ from a distribution $\mu$. The global state at the beginning of the protocol is
\[ \ket{\Psi^0} = \sum_{x,y}\sqrt{\mu(x,y)}\ket{xxyy}_{X\tilde{X} Y\tilde{Y}}\otimes\ket{\Theta^0}_{A_0B_0}\]
where the registers $\tilde{X}$ and $\tilde{Y}$ purify $X$ and $Y$ and are inaccessible to either party. In the $t$-th round of the protocol, if $t$ is odd, Alice applies a unitary $U_t: A'_{t-1}C_{t-1} \to A'_tC_t$, on her input, her memory register $A'_{t-1}$ and the message $C_{t-1}$ from Bob in the previous round (where $A'_0=XA_0$ and $C_0$ is empty), to generate the new message $C_t$, which she sends to Bob, and new memory register $A'_t$. Similarly, if $t$ is even, then Bob applies the unitary $U_t: B'_{t-1}C_{t-1} \to B'_tC_t$ and sends $C_t$ to Alice. It is easy to see that $B'_t=B'_{t-1}$ for odd $t$, and $A'_t=A'_{t-1}$ for even $t$.
The global state at the $t$-th round is then
\[ \ket{\Psi^t} = \sum_{x,y}\sqrt{\mu(x,y)}\ket{xxyy}_{X\tilde{X} Y\tilde{Y}}\otimes\ket{\Theta^t_{xy}}_{A_tB_tC_t}.\]

The quantum communication cost of this protocol $\Pi$ can be seen to be $\sum_t\log|C_t|$. The quantum information complexity of $\Pi$ with respect to distribution $\mu$ is defined as
\[ \QIC(\Pi,\mu) = \sum_{t \text{ odd}}I(\tilde{X}\tilde{Y}:C_t|YB'_t)_{\Psi^t} + \sum_{t\text{ even}}I(\tilde{X}\tilde{Y}:C_t|XA'_t)_{\Psi^t}.\]
The distributional quantum information cost with error $\eps$ error $\QIC_\eps(f,\mu)$ of a relation $f$ if the minimum quantum information cost of a protocol that computes $f$ with error $\eps$ against $\mu$, and the quantum information cost of $f$ is the maximum $\QIC_\eps(f,\mu)$ over all $\mu$. Quantum information cost satisfies $\QIC_\eps(f) \leq \Q^{\CC,*}(f)$.

The following relationship between $\QIC$ and amortized quantum communication cost also holds.
\begin{theorem}[\cite{Tou15}]\label{thm:Q-amortized}
For any relation $f$, any $\eps, \eps', \delta > 0$, there exists an integer $k$ such that
\[ \frac{\Q^\CC_{\eps+\eps'}(f^k,\mu^{\otimes k})}{k} \leq \QIC_\eps(f,\mu) + \delta.\]
\end{theorem}

\paragraph{The approximate $\gamma_2$ norm.} Let $F \subseteq \X\times\Y\times\Sigma$ be a relation. We say a system of $|\X|\times|\Y|$ real matrices $\{P_c\}_{c\in \Sigma}$ approximates $F$ to error $\eps$ if the following conditions hold:
\begin{enumerate}
\item $P_c(x,y) \geq 0$ for all $c\in \Sigma, (x,y) \in \X\times\Y$
\item $\sum_{c\in \Sigma}P_c(x,y) \leq 1$ for all $(x,y)\in \X\times\Y$
\item $\sum_{c\in F(x,y)}P_c(x,y) \geq 1-\eps$ for all $(x,y) \in \Dom(F)$.
\end{enumerate}
The $\gamma_2$ norm is a matrix norm that we will define below. For a system of matrices $\{P_c\}$, we call the $\gamma_2$ norm of the system the maximum $\gamma_2$ norm of any of the matrices $P_c$ in the system. The $\eps$-approximate $\gamma_2$ norm of a communication relation $F$, denoted by $\agamma{\eps}(F)$ is the minimum $\gamma_2$ norm of a system of matrices that $\eps$-approximates $F$.

We can also define a system of matrices that approximates $F$ with respect to a specific distribution $M$ by replacing the last condition for the $P_c$ matrices by $\sum_{(x,y) \in \Dom(F)}M(x,y)\sum_{c\in F(x,y)}P_c(x,y) \geq 1-\eps$. The distributional approximate $\gamma_2$ norm $\agamma{\eps}(F,M)$ will be defined as the minimum $\gamma_2$ norm of a system of matrices that $\eps$-approximates $F$ with respect to $M$.

\begin{definition}
For an $m\times n$ real matrix $A$, its $\gamma_2$ norm is defined as
\[ \gamma_2(A) = \min_{B,C: BC=A}\Vert B\Vert_r\Vert C\Vert_c,\]
where the minimum is taken over pairs of matrices $B,C$ such that $BC = A$, $\Vert B\Vert_r$ is the maximum $2$-norm of a row of $B$, and $\Vert C\Vert_c$ denotes the maximum $2$-norm of a column of $C$.
\end{definition}

The following theorem is proved for boolean functions in \cite{LS09a}; a similar result holds for relations.
\begin{theorem}[\cite{LS09a}]
For a communication relation $F\subseteq \X\times\Y\times\Sigma$, $\Q^{\CC,*}_\eps(F) \geq \log\agamma{\eps}(F)$. Moreover, for any distribution $M$, $F\subseteq \X\times\Y\times\Sigma$, $\Q^{\CC,*}_\eps(F,M) \geq \log\agamma{\eps}(F,M)$.
\end{theorem}

%We call a subset of $\X\times\Y$ of the form $\X'\times\Y'$, where $\X'\subseteq \X$ and $\Y'\subseteq \Y$, a (combinatorial) rectangle. We shall also use rectangle to refer to a $|\X|\times|\Y|$ matrix of the norm $vu^T$, where $u\in\B^{|X|}$ and $v\in\B^{|\Y|}$. $vu^T$ has $1$s $\X'\times\Y'$, and $0$s elsewhere, where $\X'$ is the set of locations where $u$ has $1$s, and $\Y'$ is the set of locations where $v$ has $1$s.

We now define the nuclear norm of a matrix, which is closely related to the $\gamma_2$ norm.
\begin{definition}
For an $m\times n$ real matrix $A$, its nuclear norm $\nu(A)$ is
\[ \nu(A) = \min\left\{\sum_{i=1}^k\Delta_{ii}: X\Delta Y = A, X\in \{-1,+1\}^{m\times k}, Y\in \{-1,+1\}^{k\times n}, \Delta \geq 0, \Delta \text{ diagonal} \right\}.\]
\end{definition}
We call a matrix of the form $vu^T$, where $v\in \{-1,+1\}^n$, and $u\in\{-1,+1\}^m$, a sign matrix. The nuclear norm of $A$ is then the minimum value of $\sum_|w_i|$ where $w_1, \ldots, w_k$ are real numbers that satisfy $A= \sum_{i=1}^kw_iR_i$, where each $R_i$ is a sign matrix.
\begin{theorem}[\cite{Gro53,Shanotes}]\label{thm:gamma-nu}
For any real matrix $A$,
\[ \gamma_2(A) \leq \nu(A) \leq K_G\cdot\gamma_2(A),\]
where $K_G$ is Grothendieck's constant, $1.67 \leq K_G \leq 1.79$.
\end{theorem}

\subsection{Lifting distributional approximate degree to distributional approximate \texorpdfstring{$\gamma_2$}{gamma 2} norm}

% To get a lifting theorem, we presumably need to take a dual
% of a polynomial system computing a search problem;
% we want to do this for both query and communication complexity,
% then lift the query dual to a communication dual to get the
% desired lower bound. In query complexity, we have
\paragraph{Linear program for distributional approximate degree.} It is not difficult to see that a relation $f \subseteq \{0,1\}^n\times\Sigma$ has $\adeg_\delta(f,\mu) \leq d$ if and only if the following linear program has optimal value at most $\delta$. 
\begin{equation*}
\begin{array}{lrll}
\text{min}  & 1 - \displaystyle\sum_{x\in\Dom(f)}\mu(x)\sum_{c\in f(x)}p_c(x) & &\\
\text{s.t.}&\displaystyle\sum_{c\in \Sigma} p_c(x) &\le 1 &\forall x\\
           &\displaystyle\sum_{S:|S|\le d} \hat{p}_c(S)\chi_S(x)&=p_c(x)&\forall x,c\\
           & p_c(x)&\ge 0 &\forall x,c,
\end{array}
\end{equation*}
where $\chi_S(x)$ is the $S$-th Fourier basis function for $S\subseteq [n]$, which can be thought of as a monomial.\footnote{If $x$ is thought of as a string in $\{-1,+1\}^n$ instead of $\B^n$, then $\chi_S(x)$ is the monomial $\prod_{i\in S}x_i$. We can go from the $0/1$ representation to the $\pm 1$ representation by taking $x_i \mapsto 1-2x_i$. $\prod_{i\in S}x_i$ does not remain a monomial under this mapping, but it is a polynomial of the same degree.}
Here the variables are $\eps$, $p_c(x)$ (for every $c$ and $x$), and
$\hat{p}_c(S)$ (for all $c$ and $S$); the terms $\chi_S(x)$ are constants for the purposes of the optimization.

Now let us take the dual of the linear program. The dual will have variables for the constraints:
we will have non-negative variables $\lambda(x)\ge 0$ for all $x$, and unbounded variables $\phi_c(x)$
for all $x$ and $c$. The constraints for the variables $\hat{p}_c(S)$
are $\sum_x \phi_c(x)\chi_S(x)=0$, or $\langle \phi_c,\chi_S\rangle =0$
for all $c$ and all $S$ with $|S|\le d$; this is saying each $\phi_c$
has pure high degree at least $d+1$. The constraints
for the variables $p_c(x)$ are $\lambda(x)\ge f(c,x)\mu(x)+\phi_c(x)$
for each $c$ and $x$, where $f(c,x)=1$ if $c\in f(x)$ and $f(c,x)=0$
otherwise. Finally, the objective value is to maximize
$\sum_x \mu(x)-\sum_x\lambda(x) =1-\sum_x\lambda(x)$. We get
\begin{equation*}
\begin{array}{lrll}
\text{max}  & 1-\displaystyle\sum_x\lambda(x)& &\\
\text{s.t.} &\langle \chi_S,\phi_c\rangle &= 0 &\forall c,\forall S:|S|\le d\\
            &\lambda(x)-f(c,x)\mu(x)&\ge \phi_c(x)&\forall x,c\\
            & \lambda(x)&\ge 0 &\forall x.
\end{array}
\end{equation*}

\paragraph{Linear program for distributional approximate $\gamma_2$ norm.} 

% A system of matrices $\{P_c\}$ approximates $F$ to error $\eps$
% if the usual conditions hold entrywise; the approximate
% rank of the system is the maximum approximate rank of a matrix
% $P_c$. The approximate rank of $P_c$ will be defined as the minimum
% total weight $\sum_R |w_c(R)|$ such that $P_c=\sum_R w_c(R) R$,
% where $R$ ranges over all sign matrices. Thus the approximate rank 
% to error $\eps$ is
We now give a linear program for the nuclear norm of a system of matrices (i.e., the maximum rank of any matrix in the system) that $\eps$-approximates a communication relation $F$ w.r.t distribution $M$. By \thm{gamma-nu}, up to a factor of $K_G$, this is also a linear program for $\agamma{\eps}(F)$.
\begin{equation*}
\begin{array}{lrll}
\text{min}  & T & &\\
\text{s.t.}&\displaystyle\sum_{c\in\Sigma} P_c(x,y) &\le 1 &\forall (x,y) \in \X\times\Y\\
           &\displaystyle\sum_{(x,y)\in\Dom(F)}M(x,y)\sum_{c\in F(x,y)}P_c(x,y)&\ge 1-\eps & \\
           &\displaystyle\sum_{R} w_c(R)R(x,y)&=P_c(x,y)&\forall (x,y)\in\X\times\Y,c\in\Sigma\\
           &-w_c(R), \;\; w_c(R)&\le u_c(R) &\forall c\in\Sigma,R\\
           &\displaystyle \sum_R u_c(R)& \le T &\forall c\in\Sigma\\
           & P_c(x,y)&\ge 0 &\forall (x,y)\in\Sigma,c\in\Sigma.
\end{array}
\end{equation*}
Here $R$ ranges over all $|\X|\times|\Y|$ sign matrices; $u_c(R)$ represents $|w_c(R)|$.
The variables are
$u_c(R)$ and $w_c(R)$ for all $(c,R)$, $P_c(x,y)$ for all $(c,x,y)$, and $T$;
$R(x,y)$ and $\eps$ are constants.

We now take the dual of this linear program. The dual will have variables for the constraints: variables
$\Lambda(x,y)\ge 0$ for each $x,y$ for the first constraint, a scalar $\alpha$ for the second constraint, unbounded variables $\Phi_c(x,y)$
for each $(c,x,y)$, a pair of variables $\beta_c^{-}(R)\ge 0$ and
$\beta_c^{+}(F)\ge 0$ for each $(c,R)$ corresponding to
$-w_c(R)\le u_c(R)$ and $w_c(R)\le u_c(R)$ respectively,
and variables $\eta_c\ge 0$ for each $c$.

The constraint for the variable $T$ will be $\sum_c \eta_c=1$,
forcing $\eta_c$ to be a probability distribution over $c$.
The constraints for $u_c(R)$ will be $\eta_c=\beta_c^{-}(R)+\beta_c^{+}(R)$
for each $(c,R)$,
the constraints for $w_c(R)$ will be $\sum_{x,y} R(x,y)\Phi_c(x,y)=\beta_c^{-}(R)-\beta_c^{+}(R)$
for each $(c,R)$, and the constraints for $P_c(x,y)$ will be
$\Lambda(x,y)-\alpha F_c(x,y)M(x,y)\ge \Phi_c(x,y)$ for all $(c,x,y)$, where $F_c(x,y)$
is $1$ if $c\in F(x,y)$ and $0$ otherwise.
The objective value is to maximize
$(1-\eps)\alpha\sum_{(x,y)\in\Dom(F)}M(x,y)-\sum_{x,y} \Lambda(x,y) = (1-\eps)\alpha-\sum_{x,y} \Lambda(x,y)$.

Note that 
%relaxing $\sum_c\eta_c=1$ to $\sum_c\eta_c\le 1$
% corresponds to placing the constraint $T\ge 0$ in the primal
% problem, and hence does not change the optimal value;
% similarly,
relaxing $\eta_c\ge \beta_c^{-}(R)+\beta_c^{+}(R)$
corresponds to setting $u_c(R)\ge 0$ in the primal, which
does not change the optimal value.
Moreover, the constraints on $\Phi_c,\Lambda,\alpha$ do not change
if we change $\beta_c^{-}(R)$ and $\beta_c^{+}(R)$ in a way
that preserves their difference. We can therefore decrease
both $\beta_c^{+}(R)$ and $\beta_c^{-}(R)$ for each $(c,R)$
until one of them becomes $0$. Letting $\beta_c(R)=\beta_c^{-}(R)-\beta_c^{+}(R)$,
we then get $\beta_c^{-}(R)+\beta_c^{+}(R)=|\beta_c(R)|$.
One of the constraints becomes $\langle R,\Phi_c\rangle =\beta_c(R)$
for all $c,R$; substituting this in for $\beta_c(R)$,
the other constraint involving $\beta$ becomes
$\eta_c\ge \max_R |\langle R,\Phi_c\rangle|$ for all $c$.
Therefore, we can write the dual linear program as
\begin{equation*}
\begin{array}{lrll}
\text{max}  & (1-\eps)\alpha-\langle J, \Lambda\rangle & &\\
\text{s.t.}&\Lambda-\alpha F_c\circ M&\ge \Phi_c &\forall c\\
           &\displaystyle\sum_c \max_R|\langle R,\Phi_c\rangle|&\le 1 & \\
           & \Lambda, \nu &\ge 0, &
\end{array}
\end{equation*}
where $J$ denotes the all $1$s matrix.

For sign matrices $R$ and all $|\X|\times|\Y|$ matrices $\Phi_c$,
\[\max_R|\langle R,\Phi_c\rangle| \le \|\Phi_c\|_{sp}\sqrt{|\X|\cdot|\Y|},\]
where $\Vert .\Vert_{sp}$ is the spectral norm.
Using the dual program for the nuclear norm and \thm{gamma-nu}, we can therefore lower bound the approximate $\gamma_2$ norm for a search problem by
\begin{equation}\label{eq:gamma2-lb}
\agamma{\eps}(F)\ge \frac{1}{K_G}
\max_{\substack{\nu,\Lambda\ge 0,\{\Phi_c\}_c:
\Lambda-\alpha F_c\circ M\ge \Phi_c\\}} \frac{(1-\eps)\alpha - \langle J, \Lambda\rangle}
{\sqrt{|\X|\cdot|\Y|} \sum_c\|\Phi_c\|_{sp}}.
\end{equation}

With the above definitions, we shall the following result by Sherstov to lift approximate degree to approximate $\gamma_2$ norm for a search problem.

\begin{theorem}[Pattern matrix method \cite{She11}]\label{thm:pattern}
Let $\psi\colon \B^n\to\bR$ be any real-valued function
(generally representing a dual polynomial).
Let $G\colon\B^2\times\B^2\to\B$ be the function
$G(x_1,x_2,y_1,y_2)=x_1\oplus y_1\oplus(x_2\wedge y_2)$,
and let $\psi\circ G^n$ be the $2^{2n}$ by $2^{2n}$
matrix we get by applying $G^n$ to strings $x,y\in\B^{2n}$
and applying $\psi$ to the resulting $n$-bit string.

Then the singular values of the matrix $\psi\circ G^n$
are exactly the values $2^{2n-|S|/2}|\hat{\psi}(S)|$
repeated $2^{|S|}$ times each, for every $S\subseteq[n]$.
Here $\hat{\psi}(S)=\E[\psi(x)\chi_S(x)]$ is the Fourier
coefficient of $S$ in $\psi$.
In particular, letting $\Psi=2^{-3n}\psi\circ G^n$, we have
\begin{enumerate}
\item For any $f\colon\B^n\to\bR$,
\[\langle f\circ G^n,\Psi\rangle=\langle f,\psi\rangle.\]
In particular,
\[\|\Psi\|_1=\langle \sgn(\Psi),\Psi\rangle
=\sum_x|\psi(x)|=\|\psi\|_1,\]
\[\|\Psi\|_F^2=\langle \Psi,\Psi\rangle =2^{-2n}\E[\psi(x)^2].\]
\item The spectral norm of $\Psi$ satisfies
\[\|\Psi\|_{sp}=2^{-n}\max_S |\hat{\psi}(S)|/2^{|S|/2}
\le 2^{-2n-d/2}\|\psi\|_1,\]
where $d$ is the pure-high-degree of $\psi$
(so $\hat{\psi}(S)=0$ for all $S$ with $|S|< d$).
\end{enumerate}
% \[\|\Psi\|_{F}=2^{-n}\|\hat{\psi}\|_2=2^{-n}\sqrt{\E[\psi(x)^2]},\]
% \[\|\Psi\|_1=\|\psi\|_1=\sum_x |\psi(x)|,\]
% \[\|\Psi\|_{sp}=2^{-n}\max_S |\hat{\psi}(S)|/2^{|S|/2}.\]
% If $\psi$ is a dual polynomial with the properties that
% it has pure-high-degree $d$ (so $\hat{\psi}(S)=0$ for $|S|<d$)
% and $\sum_x|\psi(x)|=1$, then
% \[\|\Psi\|_{sp}\le 2^{-2n-d/2}.\]
% For any $f\colon\B^n\to\bR$, we also have
% \[\langle f\circ G,\Psi\rangle =\langle f,\psi\rangle.\]
\end{theorem}

\begin{theorem}\label{thm:gamma2-lift}
Let $f$ be any query TFNP problem, $\mu$ be a distribution on its inputs, and $\delta>0$ be a success probability
such that $\adeg_{1-\delta}(f,\mu)\gg \polylog n$ and such that
$\delta > 2^{-\adeg_{1-\delta}(f)/3}$. Then for the gadget $G$ from \thm{pattern} $M=2^{-3n}\mu\circ G^n$ is a distribution on $\B^{2n}\B^{2n}$. Moreover, we have,
\[\log\agamma{1-2\delta}(f\circ G^n,M)\ge \Omega(\adeg_{1-\delta}(f,\mu)).\]
\end{theorem}

\begin{proof}
First we'll argue $M$ is a distribution. It is easy to see that the dimensions of $M$ are $2^{2n}\times 2^{2n}$. $\mu\circ G^n$ assigns weight $\mu(z)$ to each $(x,y) \in \B^{2n}\times\B^{2n}$ such that $G^n(x,y) = z$. So we only need to argue that for each $z$, there are $2^{3n}$ $(x,y)$ such that $G^n(x,y)$. Since $G$ has XORs, for each $z_i$, half of the inputs in the $i$-th block of $(x,y)$ should evaluate to $z_i$. Since each $G$ has $2^4$ inputs, and there are $n$ blocks of $G$-s, we get that $2^{-3n}\mu\circ G^n$ is normalized. 

We will use \eq{gamma2-lb} with $\Phi_c=2^{-3n}\phi_c\circ G^n$, with $\phi_c$ being the variable from the dual LP for $\adeg_{1-\delta}(f)$. If $\adeg_{1-\delta}(f)=d$, then $\phi_c$ is pure-high-degree $d$, and we claim $\sum_x|\phi_c(x)| \leq 4\delta$ for each $c$. To see this, note that for all $c$, $\langle \chi_\emptyset, \phi_c\rangle = \langle 1, \phi_c \rangle = \sum_x \phi_c(x)=0$. This means $\sum_x|\phi_c(x)| = 2\sum_{x: \phi_c(x)> 0}\phi_c(x)$. Moreover, for all $c$, $\phi_c(x) \leq \lambda(x) - f(c,x)\mu(x) \leq \lambda(x)$, since $f(c,x)$ and $\mu(x)$ are both non-negative. Since $\lambda(x)$ is also non-negative, and $1-\sum_x\lambda(x)$ is the objective value, we have, $\sum_{x: \phi_c(x)>0} \phi_c(x) \leq \sum_x\lambda(x) \leq 2\delta$. Therefore, by
\thm{pattern}, $\sum_c\Vert\Phi_c\Vert_{sp} \leq 2^{-2n-d/2}\sum_c\Vert\phi_c\Vert_1 \leq 2^{-2n-d/2}\cdot 4\delta|\Sigma|$, where $\Sigma$ is the certificate set for $f$. If we also use the lifted function $2^{-3n}\lambda\circ G$ with the corresponding variable $\lambda$ in the dual $\adeg_{1-\delta}(f)$ LP for $\Lambda$, and set $\alpha = 1$, we get $(1-2\delta)\alpha-\langle J,\Lambda\rangle = (1-2\delta) - \sum_z\lambda(z)$, where $z$ is the $n$-bit string we get by acting $G$ on $x$ and $y$. Moreover, if $\lambda(z) - \mu f(c,z)\mu(z) \geq \phi_c(z)$ for all $z$, then by the definition of $M$, $\Lambda(x,y) - \alpha F_c(x,y)\circ M(x,y) \geq \Phi_c$.

Putting everything into \eq{gamma2-lb} we get,
\[\agamma{1-2\delta}(F)\ge
\frac{2^{d/2}}{K_G\cdot 4\delta|\Sigma|} \left(2\delta-\sum_z \lambda(z)\right) = \frac{2^{d/2}}{4K_G\cdot \delta|\Sigma|}(2\delta-\delta) = \frac{2^{d/2}}{4K_G\cdot |\Sigma|},\]
where we have used that $\sum_z\lambda(z) = \delta$, since we are using the optimal solution for the dual LP for $\adeg_{1-\delta}(f,\mu)$.
% If we use feasible solutions to the dual of polynomial degree with error
% $\eps'$, then
% $\sum_x\mu(x)=1$, and $\sum_x\nu(x)=1-\eps'$.
Writing $2^{\polylog n}$ for the size of the certificate set $\Sigma$,
the lower bound becomes
\[\agamma{\eps}(F)\ge 2^{d/2-\polylog n}.\]
% Switching to success probabilities $\delta=1-\eps$ and $\delta'=1-\eps'$,
% we can write this as 
% \[\log \arank_{1-\delta}(f\circ G)\ge (1/2)\adeg_{1-\delta'}(f)-\polylog(n)
% -\log 1/(\delta-\delta').\]
% In particular, we have
% \[\log\arank_{1-2\delta}(f\circ G)\ge (1/2)\adeg_{1-\delta}(f)-\polylog(n)
% -\log 1/\delta.\]
% The loss is a factor of $2$ in the success probability
% and an additive $\log 1/\delta$
% as well as an additive log-number-of-certificates.
Taking logs gives us the required lower bound.
\end{proof}

A worst-case version of the above lifting theorem also holds, which we show in \app{gamma2-worst}.

\subsection{Lower bound on quantum information cost}

In this section, we will deal with communication tasks on inputs $x^1\ldots x^m$ and $y^1\ldots y^m$, where each $x^i$ or $y^i$ is a $2n^2\log(n/2)$-bit string. Each $x^i, y^i$ is divided into $n^2\log(n/2)$ blocks of size $2$, and we call the bits in the $j$-th block $x^i_{j,1}, x^i_{j,2}$. We use $z^i_j$ to denote $G(x^i_{j,1},x^i_{j,2},y^i_{j,1},y^i_{j,2})$, where $G$ is the gadget from \thm{pattern}, and $z^i$ to denote $G^{n^2\log(n/2)}(x^i,y^i)$.

% Step 1 and \thm{gamma2-lift} give us the following corollary.
% \begin{corollary}
% Consider the following communication task: Alice and Bob are given inputs $x^1\ldots x^m$ and $y^1\ldots y^m$, where each $x^i$ or $y^i$ is $2n^2$ bits, and are required to find for each $i$, a $j$ such that $z^i_j=1$, with the promise that for each $i$, there are between $n$ and $3n$ such $j$-s. Then $\agamma{1-2^{-m/10}}$ for this task is $2^{\Omega(m\sqrt{n})}$.
% \end{corollary}

% Using this corollary, we can prove the following lemma in the same way \thm{q-direct-prod} is proved, noting that when $x^i, y^i$ are are sampled uniformly at random, $z^i$ is also uniformly random.
% \begin{lemma}
% Let $(\AndOr_n\circ G^{n^2\log(n/2)})^m$ denote the problem of solving $m$ separate copies of the communication search problem $\AndOr_n\circ G^{n^2\log(n/2)}$. Then there is a universal constant $C$ such that when the input is sampled from the uniform distribution, any quantum communication protocol that solves this problem with probability at least $2^{-\frac{m}{C\log n}}$ must have $\Omega(m\sqrt{n}/\log n)$ communication.
% \end{lemma}

% \begin{proof}
% $(\AndOr_n\circ G^{n^2\log(n/2)})^m = \AndOr^m_n\circ G^{m\cdot n^2\log(n/2)}$. From \thm{q-direct-prod}, $\adeg_{1-2^{-m/C\log n}}(\AndOr_n^m, U_{mn^2\log(n/2)}) = \Omega(m\sqrt{n}/\log n)$, where $U_{mn^2\log(n/2)}$ is the uniform distribution on $\B^{mn^2\log(n/2)}$. Since distributional approximate $\gamma_2$ norm lower bounds $\Q^{\CC,*}$ with the same distribution, we get the required lower bound.
% \end{proof}

\begin{theorem}\label{thm:QIC-lb}
For $t=\log^5 n$, $\QIC(\AndOr_{\clH,t}\circ G^{tn^2\log(n/2)}) = \Omega(\sqrt{n}\log^4n)$.
\end{theorem}
\begin{proof}
We will lower bound the $\Q^{\CC,*}((\AndOr_{\clH,t}\circ G^{tn^2\log(n/2)})^k)$ for arbitrarily large $k$. By \thm{Q-amortized}, this will lower bound $\QIC$. Note that $\AndOr_{\clH,t}\circ G^{tn^2\log(n/2)}(x,y)$ amounts to solving $\AndOr_n^t$ on $(z^1\oplus h, \ldots z^t\oplus h)$ where $z^i = G^{n^2\log(n/2)}(x^i,y^i)$, for any $h$ in $\clH$, and $(\AndOr_{\clH,t}\circ G^{tn^2\log(n/2)})^k$ is the $k$-fold parallel repetition of this.

For fixed $h^1,\ldots, h^k \in \clH^k$, if $z^1,\ldots, z^{kt}$ are distributed uniformly on $\B^{ktn^2\log(n/2)}$, then $(z^1\oplus h^1, \ldots, z^t\oplus h^1), \ldots, (z^{kt-t+1}\oplus h^k, \ldots, z^{kt}\oplus h^k)$ are also distributed uniformly on $\B^{ktn^2\log(n/2)}$. \thm{d-direct-prod} tells us that any polynomial system of degree $o(kt\sqrt{n}/\log n)$ computes $\AndOr_n^{kt}$ with success probability at most $2^{-kt/C\log n}$ on $(z^1\oplus h^1, \ldots, z^t\oplus h^1), \ldots, (z^{kt-t+1}\oplus h^k, \ldots, z^{kt}\oplus h^k)$ when the $z^i$-s are uniform. By \thm{gamma2-lift}, any system of matrices of $\gamma_2$ norm  $o(kt\sqrt{n}/\log n)$ approximates $\AndOr_n^{kt}$ on $(z^1\oplus h^1, \ldots, z^t\oplus h^1), \ldots, (z^{kt-t+1}\oplus h^k, \ldots, z^{kt}\oplus h^k)$, where $z^1, \ldots, z^{kt} = G^{ktn^2\log(n/2)}(x,y)$ with probability at most $2\cdot 2^{-kt/C\log n}$ over uniform $x,y$ (since the uniform distribution lifted with $G$ is also a uniform distribution). Consequently, any quantum protocol of communication $o(kt\sqrt{n}/\log n)$ approximates $\AndOr_n^{kt}$ on $(z^1\oplus h^1, \ldots, z^t\oplus h^1), \ldots, (z^{kt-t+1}\oplus h^k, \ldots, z^{kt}\oplus h^k)$, where $z^1, \ldots, z^{kt} = G^{ktn^2\log(n/2)}(x,y)$ with probability at most $2\cdot 2^{-kt/C\log n}$ over uniform $x,y$. We know $|\clH| = 2^{2\log^3 n\log\log(n/2)}$, and the quantum protocol succeeds at computing $(\AndOr_{\clH,t}\circ G^{n^2\log(n/2)})^k$ if it can solve $\AndOr_n^{kt}$ on $(z^1\oplus h^1, \ldots, z^t\oplus h^1), \ldots, (z^{kt-t+1}\oplus h^k, \ldots, z^{kt}\oplus h^k)$ for any $h^1, \ldots, h^k \in \clH^k$. Therefore, by the union bound, the success probability of a quantum protocol with communication at most $o(kt\sqrt{n}/\log n)$ at this task is $2^{2k\log^3n\log\log(n/2)}\cdot 2^{-k\log^4 n/C}$.

Letting $U$ denote the uniform distribution over $\B^{2tn^2\log(n/2)}\times\B^{2tn^2\log(n/2)}$, for large enough $k$, we have by \thm{Q-amortized},
\begin{align*}
\QIC(\AndOr_{\clH,t}\circ G^{tn^2\log(n/2)}) & \geq \QIC(\AndOr_{\clH,t}\circ G^{tn^2\log(n/2)}, U) \\
& \geq \frac{\Q^{\CC,*}((\AndOr_{\clH,t}\circ G^{tn^2\log(n/2)})^k,U^k)}{k} - 1 \\
& = \Omega(\sqrt{n}\log^4(n)). \qedhere
\end{align*}
\end{proof}

\begin{theorem}\label{thm:QIC-rprt}
There exists a family of $\mathsf{TFNP}$ problems $f$ for communication, with inputs in $\B^{\poly(n)}\times\B^{\poly(n)}$ such that
\[
\QIC(f) = \Omega(\sqrt{n}\log^5n), \quad \rprt(f) = \polylog(n).
\]
\end{theorem}
\begin{proof}
$f$ will be $\AndOr_{\clH,t}\circ G^{tn^2\log(n/2)}$. The lower bound for $\QIC$ follows from \thm{QIC-lb}. The upper bound for $\rprt$ follows in the exact same way as \thm{rprt-ub}, except using the gadget $G$ instead of $\IP$.
\end{proof}

\phantomsection\addcontentsline{toc}{section}{References} %Adds References to TOC
\renewcommand{\UrlFont}{\ttfamily\small}
\let\oldpath\path
\renewcommand{\path}[1]{\small\oldpath{#1}}
\emergencystretch=1em % Helps remove overfull hboxes in bibliography
\printbibliography

\appendix

\section{Worst-case degree to worst-case approximate \texorpdfstring{$\gamma_2$}{gamma 2} norm lifting}\label{app:gamma2-worst}

\paragraph{Linear program for approximate degree.} A relation $f \subseteq \{0,1\}^n\times\Sigma$ has $\adeg_\delta(f) \leq d$ if and only if the following linear program has optimal value at most $\delta$. 
\begin{equation*}
\begin{array}{lrll}
\text{min}  & \eps& &\\
\text{s.t.}&\displaystyle\sum_{c\in \Sigma} p_c(x) &\le 1 &\forall x\\
           &\displaystyle\sum_{c\in f(x)}p_c(x)&\ge 1-\eps &\forall x\in \Dom(f)\\
           &\displaystyle\sum_{S:|S|\le d} \hat{p}_c(S)\chi_S(x)&=p_c(x)&\forall x,c\\
           & p_c(x)&\ge 0 &\forall x,c,
\end{array}
\end{equation*}
This LP has a separate correctness constraint for each input $x$, because it requires worst-case correctness. We need a separate variable $\eps$ for the error for each variable, and the objective function minimizes $\eps$. Here the variables are $\eps$, $p_c(x)$ (for every $c$ and $x$), and
$\hat{p}_c(S)$ (for all $c$ and $S$); the terms $\chi_S(x)$ are constants for the purposes of the optimization.

The dual of this LP will have an extra variable
$\mu(x)\ge 0$ for all $x\in\Dom(f)$ for each correctness constraint, and the dual constraint for the variable $\eps$ will
$\sum_x \mu(x)=1$. We use the same symbol for these dual variables as the probability distribution in the original definition, because $\mu$ will be constrained to be a probability distribution in the dual. Overall the dual is:
\begin{equation*}
\begin{array}{lrll}
\text{max}  & 1-\displaystyle\sum_x\lambda(x)& &\\
\text{s.t.}&\displaystyle\sum_{x\in\Dom(f)} \mu(x) &= 1 &\\
           &\langle \chi_S,\phi_c\rangle &= 0 &\forall c,\forall S:|S|\le d\\
           &\lambda(x)-f(c,x)\mu(x)&\ge \phi_c(x)&\forall x,c\\
           & \mu(x),\lambda(x)&\ge 0 &\forall x.
\end{array}
\end{equation*}

\paragraph{Linear program for approximate $\gamma_2$ norm.} 

% A system of matrices $\{P_c\}$ approximates $F$ to error $\eps$
% if the usual conditions hold entrywise; the approximate
% rank of the system is the maximum approximate rank of a matrix
% $P_c$. The approximate rank of $P_c$ will be defined as the minimum
% total weight $\sum_R |w_c(R)|$ such that $P_c=\sum_R w_c(R) R$,
% where $R$ ranges over all sign matrices. Thus the approximate rank 
% to error $\eps$ is
Similarly adding worst-case correctness constraints to the optimization problem for the nuclear norm, we get the following linear program for $\agamma{\eps}(F)$ (up to a factor of $K_G$):
\begin{equation*}
\begin{array}{lrll}
\text{min}  & T & &\\
\text{s.t.}&\displaystyle\sum_{c\in\Sigma} P_c(x,y) &\le 1 &\forall (x,y) \in \X\times\Y\\
           &\displaystyle\sum_{c\in F(x,y)}P_c(x,y)&\ge 1-\eps &\forall (x,y)\in \Dom(F)\\
           &\displaystyle\sum_{R} w_c(R)R(x,y)&=P_c(x,y)&\forall (x,y)\in\X\times\Y,c\in\Sigma\\
           &-w_c(R), \;\; w_c(R)&\le u_c(R) &\forall c\in\Sigma,R\\
           &\displaystyle \sum_R u_c(R)& \le T &\forall c\in\Sigma\\
           & P_c(x,y)&\ge 0 &\forall (x,y)\in\Sigma,c\in\Sigma.
\end{array}
\end{equation*}

The dual of this will have extra variables for the constraints: $M'(x,y)\ge 0$ for each $(x,y)\in\Dom(F)$
for each correctness constraint. This $M'$ is not a probability distribution, but can be thought of as a probability distribution multiplied by the variable $\alpha$ from the dual for the distributional version. $M'$ can also take non-zero values outside of $\Dom(F)$. The objective value is to maximize
$\sum_{x,y} \Lambda(x,y)-(1-\eps)\sum_{(x,y)\in\Dom(F)}M'(x,y)$. Overall the dual looks like:

%Note that 
%relaxing $\sum_c\eta_c=1$ to $\sum_c\eta_c\le 1$
% corresponds to placing the constraint $T\ge 0$ in the primal
% problem, and hence does not change the optimal value;
% similarly,
% relaxing $\eta_c\ge \beta_c^{-}(R)+\beta_c^{+}(R)$
% corresponds to setting $u_c(R)\ge 0$ in the primal, which
% does not change the optimal value.
% Moreover, the constraints on $\Phi_c,\Lambda,M$ do not change
% if we change $\beta_c^{-}(R)$ and $\beta_c^{+}(R)$ in a way
% that preserves their difference. We can therefore decrease
% both $\beta_c^{+}(R)$ and $\beta_c^{-}(R)$ for each $(c,R)$
% until one of them becomes $0$. Letting $\beta_c(R)=\beta_c^{-}(R)-\beta_c^{+}(R)$,
% we then get $\beta_c^{-}(R)+\beta_c^{+}(R)=|\beta_c(R)|$.
% One of the constraints becomes $\langle R,\Phi_c\rangle =\beta_c(R)$
% for all $c,R$; substituting this in for $\beta_c(R)$,
% the other constraint involving $\beta$ becomes
% $\eta_c\ge \max_R |\langle R,\Phi_c\rangle|$ for all $c$.
% Therefore, we can write the dual linear program as
\begin{equation*}
\begin{array}{lrll}
\text{max}  & (1-\eps)\langle \Dom(F), M'\rangle-\langle J, \Lambda\rangle & &\\
\text{s.t.}&\Lambda-F_c\circ M'&\ge \Phi_c &\forall c\\
           &\displaystyle\sum_c \max_R|\langle R,\Phi_c\rangle|&\le 1 & \\
           & \Lambda,M' &\ge 0, &
\end{array}
\end{equation*}
where we use $\Dom(F)$
to represent the matrix with $1$ for $(x,y)\in\Dom(F)$ and $0$ elsewhere.

For a total function, $\Dom(F)=J$. Using the same bound on $\max_R|\langle R,\Phi_c\rangle|$ as earlier, we can say
\[
\agamma{\eps}(F)\ge \frac{1}{K_G}
\max_{\substack{M',\Lambda\ge 0,\{\Phi_c\}_c:
\Lambda-F_c\circ M'\ge \Phi_c\\}} \frac{\langle J, (1-\eps)M'-\Lambda\rangle}
{\sqrt{|\X|\cdot|\Y|} \sum_c\|\Phi_c\|_{sp}}.
\]

\begin{theorem}
Let $f$ be any query TFNP problem, and let $\delta>0$ be a success probability
such that $\adeg_{1-\delta}(f)\gg \polylog n$ and such that
$\delta > 2^{-\adeg_{1-\delta}(f)/3}$. Then for the gadget $G$ from \thm{pattern},
\[\log\agamma{1-2\delta}(f\circ G)\ge \Omega(\adeg_{1-\delta}(f)).\]
\end{theorem}

\begin{proof}
As before, we will use $\Phi_c=2^{-3n}\phi_c\circ G$, with $\phi_c$ being the variable from the dual LP for $\adeg_{1-\delta}(f)$. If $\adeg_{1-\delta}(f)=d$, then $\phi_c$ is pure-high-degree $d$, and $\sum_x|\phi_c(x)| \leq 4\delta$ for each $c$, using the fact that the variables $\mu(x)$ are constrained to be non-negative. Therefore, by
\thm{pattern}, $\sum_c\Vert\Phi_c\Vert_{sp} \leq 2^{-2n-d/2}\sum_c\Vert\phi_c\Vert_1 \leq 2^{-2n-d/2}\cdot 4\delta|\Sigma|$, where $\Sigma$ is the certificate set for $f$. If we also use the lifted functions $2^{-3n}\mu\circ G$
and $2^{-3n}\lambda\circ G$ with the corresponding variables in the dual $\adeg_{1-\delta}(f)$ LP for $M'$ and $\Lambda$, $\langle J, (1-2\delta)M'-\Lambda\rangle = \sum_z(1-2\delta)\mu(z) - \lambda(z)$, where $z$ is the $n$-bit string we get by acting $G$ on $x$ and $y$.

Putting everything in we get,
\[\agamma{1-2\delta}(F)\ge
\frac{2^{d/2}}{K_G\cdot 4\delta|\Sigma|} \sum_z 2\delta\mu(z)-\nu(z) = \frac{2^{d/2}}{4K_G\cdot \delta|\Sigma|}((2\delta-\delta) = \frac{2^{d/2}}{4K_G\cdot |\Sigma|}.\]
Now putting in the size of $\Sigma$ and taking logs we get the required result.
% Switching to success probabilities $\delta=1-\eps$ and $\delta'=1-\eps'$,
% we can write this as 
% \[\log \arank_{1-\delta}(f\circ G)\ge (1/2)\adeg_{1-\delta'}(f)-\polylog(n)
% -\log 1/(\delta-\delta').\]
% In particular, we have
% \[\log\arank_{1-2\delta}(f\circ G)\ge (1/2)\adeg_{1-\delta}(f)-\polylog(n)
% -\log 1/\delta.\]
% The loss is a factor of $2$ in the success probability
% and an additive $\log 1/\delta$
% as well as an additive log-number-of-certificates.
Taking logs gives us the required lower bound.
\end{proof}

\end{document}

%% file: scripts.tex
\AtEveryBibitem{
 \clearlist{address}
 \clearfield{date}
 \clearlist{location}
 \clearfield{month}
 \clearfield{series}
 \clearfield{pages}
 \clearlist{organization}
 \clearfield{number}
 \clearlist{language}

 \ifentrytype{book}{}{% Remove publisher etc. except for books
  \clearlist{publisher}
  \clearname{editor}
  \clearfield{issn}
  \clearfield{isbn}
  \clearfield{volume}
 }
}

\DeclareFieldFormat{eprint:eccc}{
\mkbibacro{ECCC}\addcolon\space\ifhyperref
    {\href{http://eccc.hpi-web.de/report/#1/}{\nolinkurl{#1}}}
    {\nolinkurl{#1}}}
\DeclareFieldAlias{eprint:ECCC}{eprint:eccc}
\DeclareFieldFormat{eprint:iacr}{Cryptology ePrint\addcolon\space\ifhyperref
    {\href{https://eprint.iacr.org/#1}{\nolinkurl{#1}}}
    {\nolinkurl{#1}}}
\DeclareFieldAlias{eprint:IACR}{eprint:iacr}
\DeclareFieldFormat{eprint:arxiv}{arXiv\addcolon\space\ifhyperref
    {\href{https://arxiv.org/abs/#1}{\nolinkurl{#1}}}
    {\nolinkurl{#1}}}

\DeclareFieldFormat[article]{title}{#1}
\DeclareFieldFormat[inproceedings]{title}{#1}
\DeclareFieldFormat[incollection]{title}{#1}
\DeclareFieldFormat[online]{title}{#1}

\renewbibmacro{in:}{}

\newcommand{\lName}{1}

\newcommand{\donothing}[1]{#1}

\newcommand{\SIDMA}{\if\lName1\donothing{{SIAM} Journal on Discrete Mathematics}\else{SIDMA}\fi}
\newcommand{\JACM}{\if\lName1\donothing{Journal of the {ACM}}\else{JACM}\fi}
\newcommand{\SICOMP}{\if\lName1\donothing{{SIAM} Journal on Computing}\else{SICOMP}\fi}
\newcommand{\ToC}{\if\lName1\donothing{Theory of Computing}\else{ToC}\fi}
\newcommand{\ToCGS}{\if\lName1\donothing{Theory of Computing Graduate Surveys}\else{ToC}\fi}
\newcommand{\TOCT}{\if\lName1\donothing{{ACM} Transactions on Computation Theory}\else{TOCT}\fi}
\newcommand{\ToIT}{\if\lName1\donothing{{IEEE} Transactions on Information Theory}\else{TOCT}\fi}
\newcommand{\CCjournal}{\if\lName1\donothing{Computational Complexity}\else{CC}\fi}
\newcommand{\CJTCS}{\if\lName1\donothing{Chicago Journal of Theoretical Computer Science}\else{CJTCS}\fi}

\newcommand{\TCS}{\if\lName1\donothing{Theoretical Computer Science}\else{TCS}\fi}
\newcommand{\IPL}{\if\lName1\donothing{Information Processing Letters}\else{IPL}\fi}
\newcommand{\JCSS}{\if\lName1\donothing{Journal of Computer and System Sciences}\else{JCSS}\fi}

\newcommand{\RSA}{\if\lName1\donothing{Random Structures and Algorithms}\else{RSA}\fi}
\newcommand{\JCTA}{\if\lName1\donothing{Journal of Combinatorial Theory, Series A}\else{JCTA}\fi}
\newcommand{\JCTB}{\if\lName1\donothing{Journal of Combinatorial Theory, Series B}\else{JCTB}\fi}
\newcommand{\PJM}{\if\lName1\donothing{Pacific Journal of Mathematics}\else{PJM}\fi}
\newcommand{\QICjournal}{\if\lName1\donothing{Quantum Information and Computation}\else{QIC}\fi}
\newcommand{\IJQI}{\if\lName1\donothing{International Journal of Quantum Information}\else{IJQI}\fi}
\newcommand{\PRA}{\if\lName1\donothing{Physical Review A}\else{PRA}\fi}
\newcommand{\PRL}{\if\lName1\donothing{Physical Review Letters}\else{PRL}\fi}
\newcommand{\VLDB}{\if\lName1\donothing{International Journal on Very Large Data Bases}\else{VLDB}\fi}

\DefineBibliographyStrings{english}{%
  backrefpage = {p{.}},% originally "cited on page"
  backrefpages = {pp{.}},% originally "cited on pages"
}

\newtheorem{theorem}{Theorem}
\newtheorem{lemma}[theorem]{Lemma}

\newtheorem{corollary}[theorem]{Corollary}
\newtheorem{definition}[theorem]{Definition}

\newtheorem{open}{Open Problem}
\newtheorem{trick}[theorem]{Trick}
\newtheorem{observation}[theorem]{Observation}
\theoremstyle{definition}

\newcommand{\eq}[1]{\hyperref[eq:#1]{(\ref*{eq:#1})}}
\renewcommand{\sec}[1]{\hyperref[sec:#1]{Section~\ref*{sec:#1}}}
\newcommand{\thm}[1]{\hyperref[thm:#1]{Theorem~\ref*{thm:#1}}}
\newcommand{\lem}[1]{\hyperref[lem:#1]{Lemma~\ref*{lem:#1}}}
\newcommand{\defn}[1]{\hyperref[def:#1]{Definition~\ref*{def:#1}}}
\newcommand{\prop}[1]{\hyperref[prop:#1]{Proposition~\ref*{prop:#1}}}
\newcommand{\cor}[1]{\hyperref[cor:#1]{Corollary~\ref*{cor:#1}}}
\newcommand{\fig}[1]{\hyperref[fig:#1]{Figure~\ref*{fig:#1}}}
\newcommand{\tab}[1]{\hyperref[tab:#1]{Table~\ref*{tab:#1}}}
\newcommand{\alg}[1]{\hyperref[alg:#1]{Algorithm~\ref*{alg:#1}}}
\newcommand{\app}[1]{\hyperref[app:#1]{Appendix~\ref*{app:#1}}}
\newcommand{\conj}[1]{\hyperref[conj:#1]{Conjecture~\ref*{conj:#1}}}
\newcommand{\chap}[1]{\hyperref[chap:#1]{Chapter~\ref*{chap:#1}}}
\newcommand{\trk}[1]{\hyperref[trick:#1]{Trick~\ref*{trick:#1}}}

\newcommand{\comment}[1]{\textup{{\color{red}#1}}}

\DeclareMathOperator{\poly}{poly}
\DeclareMathOperator{\polylog}{polylog}
\newcommand{\tO}{\widetilde{O}}
\newcommand{\tOmega}{\widetilde{\Omega}}
\newcommand{\tTheta}{\widetilde{\Theta}}

\DeclareMathOperator{\Dom}{Dom}

\newcommand{\B}{\{0,1\}}

\DeclareMathAlphabet{\mathbbold}{U}{bbold}{m}{n}

\DeclareMathOperator*{\E}{\mathbb{E}} % expectation
\DeclareMathOperator{\bR}{\mathbb{R}}
\DeclareMathOperator{\bN}{\mathbb{N}}

\DeclareMathOperator{\Id}{\mathbb{I}}

\DeclareMathOperator{\cost}{cost}

\DeclareMathOperator{\sgn}{sgn}
\DeclareMathOperator{\rank}{rank}

\newcommand{\X}{\mathcal{X}}
\newcommand{\Y}{\mathcal{Y}}

\newcommand{\clH}{\mathcal{H}}
\newcommand{\clR}{\mathcal{R}}

\newcommand{\CC}{\mathrm{CC}}
\newcommand{\QCC}{\mathrm{QCC}}
\newcommand{\eps}{\varepsilon}
\newcommand{\Bad}{\mathrm{Bad}}

\newcommand{\Or}{\mathtt{Or}}
\newcommand{\PromiseOR}{\mathtt{PromiseOR}}

\newcommand{\IP}{\mathtt{IP}}
\newcommand{\tPH}{\mathtt{PH}}
\newcommand{\AndOr}{\mathtt{AndOr}}
\DeclareMathOperator{\D}{D}
\DeclareMathOperator{\R}{R}
\DeclareMathOperator{\Q}{Q}
\DeclareMathOperator{\C}{C}

\DeclareMathOperator{\s}{s}
\DeclareMathOperator{\bs}{bs}
\DeclareMathOperator{\fbs}{fbs}
\DeclareMathOperator{\cbs}{cbs}
\DeclareMathOperator{\fcbs}{fcbs}

\DeclareMathOperator{\Adv}{Adv}
\DeclareMathOperator{\CAdv}{CAdv}

\DeclareMathOperator{\cert}{cert}
\DeclareMathOperator{\rprt}{relprt}
\DeclareMathOperator{\fcn}{fcn}

\DeclareMathOperator{\adeg}{\widetilde{\deg}}
\DeclareMathOperator{\appdeg}{\adeg\!{}^{+2}}
\newcommand{\agamma}[1]{\widetilde{\gamma}_{2,#1}}

\DeclareMathOperator{\udeg}{\widetilde{udeg}}

\DeclareMathOperator{\IC}{IC}
\DeclareMathOperator{\QIC}{QIC}